\documentclass[aps,pra,twocolumn,floatfix,longbibliography, twocolumn, superscriptaddress]{revtex4-2}
\usepackage[utf8]{inputenc}
\usepackage[english]{babel}
\usepackage[T1]{fontenc}

\usepackage[dvipsnames]{xcolor} 

\usepackage{bbm}
\usepackage{algorithm}
\usepackage{algpseudocode}
\usepackage{tikz}
\usepackage{lipsum}
\usepackage{amsmath, amssymb}
\usepackage{multirow}
\usepackage{amsthm}
\newtheorem{theorem}{Theorem}
\theoremstyle{definition}
\newtheorem{definition}{Definition}[section]

\usepackage{url}
\usepackage[colorlinks]{hyperref}
\hypersetup{%
	plainpages=true,
	breaklinks=true,
	hypertexnames=false,
	pageanchor=true,
	colorlinks=true,
	linkcolor={blue},
	citecolor={red},
	urlcolor={blue},
	anchorcolor={black}
}

\newcommand{\placeholder}[1]{}

\DeclareMathOperator*{\argmin}{arg\,min}

\begin{document}
\title{Qudit Machine Learning}
  
\author{Sebastián Roca-Jerat}
\email{sroca@unizar.es}
\author{Juan Román-Roche}
\author{David Zueco}
\email{dzueco@unizar.es}
\address{Instituto de Nanociencia y Materiales de Aragón (INMA),
  CSIC-Universidad de Zaragoza, Zaragoza 50009,
  Spain}
\address{Departamento de Física de la Materia Condensada, Universidad de Zaragoza, Zaragoza 50009, Spain}


\begin{abstract}
We present a comprehensive investigation into the learning capabilities of a simple $d$-level system (qudit). Our study is specialized for classification tasks using real-world databases, specifically the Iris, breast cancer, and MNIST datasets. We explore various learning models in the metric learning framework, along with different encoding strategies. In particular, we employ data re-uploading techniques and maximally orthogonal states to accommodate input data within low-dimensional systems. Our findings reveal optimal strategies, indicating that when the dimension of input feature data and the number of classes are not significantly larger than the qudit's dimension, our results show favorable comparisons against the best classical models. This trend holds true even for small quantum systems, with dimensions $d<5$ and utilizing algorithms with a few layers ($L=1,2$). However, for high-dimensional data such as MNIST, we adopt a hybrid approach involving dimensional reduction through a convolutional neural network. In this context, we observe that \emph{small} quantum systems often act as bottlenecks, resulting in lower accuracy compared to their classical counterparts.

\end{abstract}

\maketitle


\section{Introduction}
\label{sec:intro}

Machine learning (ML) has been integrated into our
daily routines. This broad adoption is closely linked to
advancements in hardware capabilities. Alongside its broad adoption, ML has brought attention to the need for accelerating algorithms and to the energy costs of computation, emphasizing the importance of exploring alternative computing paradigms \cite{Schuld2014, patterson2021carbon}.

The interest in using quantum processors for ML tasks stems from the hope that certain problems could be solved more quickly on a quantum computer and/or that quantum computing could be more energy-efficient, thereby reducing the environmental impact of ML \cite{auffeves2022}. However, whether these
expectations are met for a wide range of problems remains unclear \cite{havlivcek2019supervised, Hamerly2019, albrecht2023quantum, Ebadi2022, Tang2019, Tang2019a, Arrazola2020, Huang2021, Schuld2022}. Investigating these aspects is crucial to either substantiate or refute these claims.

Concurrently, there is the rise of Physical Neural Networks (PNNs), which consist of dedicated physical systems, such as nonlinear optical or mechanical resonators, manipulated to produce specific computational outputs \cite{yao2020protonic, wright2022deep,Stern2023}. To date, PNNs have  been explored within classical systems.

In this paper, we bridge the gap between quantum ML and PNNs by considering a quantum physical system as an example of a PNN. We specifically examine a simple d-level system (qudit), wherein transitions can be induced \cite{hrmo2023native, chi2022programmable, ringbauer2022universal, low2023control, gimeno2021broad, chiesa2023blueprint}, and the population of levels post-manipulation can be measured \cite{gomez2022dispersive}. Our focus is on how this system performs in classification tasks. This minimal quantum system provides an opportunity to compare the universality, expressiveness, and performance of a $d$-dimensional unitary operation against standard classical algorithms. We will explore various strategies based on dataset dimension $D_x$, qudit dimension $d$, and the number of classes $K$. The primary goal is to understand how a minimal physical system, governed by quantum mechanics, can execute ML tasks, without necessarily claiming any superiority over other systems. We believe that the concepts discussed and the results obtained may also be of use in the field of quantum-inspired algorithms, where they seek to exploit concepts and tools from quantum mechanical theory to develop potentially more efficient classical algorithms \cite{shi2020quantum, garciamolina2023global, Arrazola2020}.

\subsection{Our Work in Context}

The use of quantum processors for machine learning tasks is a rapidly growing field. It is important to note that this brief overview may not cover all relevant works in the field. 

Previous research has focused on quantum circuits with single and two-qubit gates, which are fundamental architectures in NISQ quantum computers. Various models, such as Kernel and quantum neural networks, have been explored for classification problems . Different learning and measurement protocols have been compared in these studies
\cite{ Altaisky2001,lloyd2014,  Schuld2014, Rebentrost2014,farhi2018classification,  mari2020transfer, Schuld2020, Bartkiewicz2020,schuld2021supervised, sancho2022quantum, rudolph2022generation}. See also the reviews \cite{biamonte2017quantum, perdomo2018opportunities, schuld2019, cerezo2021variational}. 

In general, most of these works address the use of quantum systems from a digital perspective, i.e. in the development of algorithms and circuits that can be implemented in any quantum computer that meets a number of basic requirements (connectivity, native gate set, etc.). However, in the spirit of the aforementioned PNNs, we are interested in a more analog-like perspective and there are seminal works that align more closely with the philosophy of our current research. These studies investigate the capabilities of the smallest quantum systems, \emph{i.e.} two level systems and test their performance. They also introduce the concept of data re-uploading \cite{perez2020data}, a concept that we will discuss in this work as well. Reference \cite{perez2021one} demonstrated that a two-level quantum system with data re-uploading can serve as a universal approximator for any function. This is achieved by leveraging results from quantum Fourier analysis, which have been recently generalized to the qudit scenario \cite{casas2023multi}. Moreover, the generalization of data re-uploading to qudits has been discussed in \cite{wach2023data}, where the qudit dimension is fixed to the number of classes. The implementation of classification techniques on a superconducting qutrit carried out in \cite{cao2023encoding} also proves to be of great interest for the context of the present work. Here we generalise to a completely arbitrary case in terms of number of levels and number of classes to be classified.

\emph{
Our aim here is to test the learning capabilities of a single quantum system with $d$ levels (qudit), conducting an extensive study that employs various types of learning, encoding, and measurement algorithms.
}

\subsection{Main Results and Organization}

In this work, we take a broader perspective on learning with qudits within the metric learning paradigm \cite{lloyd2020quantum}, discussing both implicit and explicit approaches \cite{nghiem2021unified}, each with its own advantages and drawbacks. We explore different types of encodings, such as data re-uploading and the use of maximally orthogonal states, to accommodate the dimension of the data set within few-level quantum machines.

Regarding the results, we discuss various prototypical datasets with unique characteristics that allow us to explore the efficiency of the different methods discussed. Additionally, we test the potential of hybrid classical-quantum models suitable for cases with high-dimensional data sets.

For data sets with dimensions and  a number of classes not significantly larger than the dimension of the qudits, we provide results that compare favorably to the best classical models. In such cases, a few-layer protocol with a small quantum system of dimension $d\sim 3-5$ suffices. However, for higher-dimensional data sets, hybrid models perform acceptably, but the small dimension of the quantum system cannot outperform advanced classical techniques like convolutional neural networks.

The rest of the manuscript is organized as follows: in section \ref{sec:learningwq} we present our theoretical backbone where we cover the control of our physical system, the optimisation method and how data is encoded. In section \ref{sec:results} we test these tools using different datasets and different models according to the needs of each problem as well as discussing the computational complexity. Finally, in section \ref{sec:Dissipation} we check what effect the main sources of noise have on our physical system in both the training and testing phases. We leave the conclusions and final remarks for section \ref{sec:conclusions}.

\section{Learning with qudits}\label{sec:learningwq}

\begin{figure*}[t!]
    \centering
    \includegraphics[width = \textwidth]{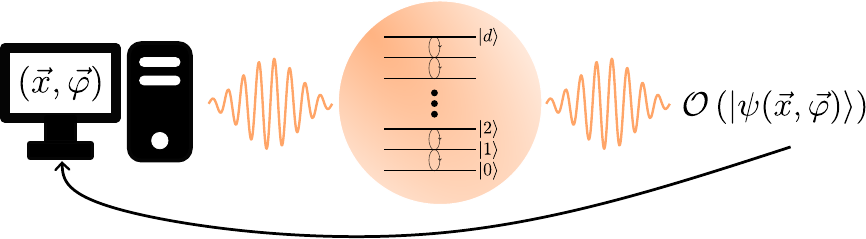}
    \caption{Variational optimization process on a d-level system. Classical devices are used to perform the optimization task and to control the operations to implement on the physical system (the qudit). After coherently controlling the system, we perform some measurements that allows us to extract the information need to compute the objective function and iterate the optimization process.}
    \label{fig:vqaqudit}
\end{figure*}

Our focus is on classifiers based on quantum variational algorithms (QVA) for addressing supervised learning problems.
In the supervised learning scenario, the data set consists of input vectors containing the feature values and their corresponding labels, 
$\{\vec x_i, y_i\}_{i=1}^N$   
with
${\rm dim }(\vec x_i)= D_x \; \& \; y_i=1,..,K $.
We are assuming that data can be classified in $K$-classes.
Besides, the data is splitted in training ($N_{training}$) and testing ($N_{test}$). However,  for certain problems it is convenient to further split the training data set into validation ($N_{validation}$), which is used to assess that no overfitting is happening and proper training (the one with which the parameters of the model are optimized).
\\

Typically, QVAs can be divided into three steps: \emph{i)} encoding the data onto a wave function $\vec x_i \to | \psi (\vec x_i) \rangle$, \emph{ii)} evolving it to a final state $ U (\vec \varphi)| \psi (\vec x_i) \rangle $. However, these two steps can be encapsulated in a step that depends on the input data set and the parameters $\vec \varphi$, namely,
\begin{equation}
\label{Uxphi}
|\psi (\vec x_i, \vec \varphi)\rangle = \hat{U}(\vec{x}_i; \vec{\varphi})|0\rangle \; .
\end{equation}
The final part of the algorithm, \emph{iii)} involves performing a measurement $o_i \equiv \langle 0 | U_\theta^\dagger (\vec x_i, \vec \varphi)   {\mathcal  O} U_\theta (\vec x_i, \vec \varphi) |0 \rangle$, which can be estimated up to a precision of $1/\sqrt{N_{\rm shots}}$ in any measurement-based experimental device \cite{nielsen2002quantum}, being $N_{\rm shots}$ the number of measurements performed to obtain the statistics of the observable $o_i$. This is the general scenario; however, in the next section, we provide explicit algorithms together with their complexity estimations for both implicit and explicit learning with qudits. This process undergoes an iterative procedure, as illustrated in Figure \ref{fig:vqaqudit}. By optimizing specific parameters ($\vec \varphi$), our objective is to minimize a predetermined objective function that depends on the output $o_i$, often referred to as the cost/loss function, which quantifies the level of misclassification in our model \cite{biamonte2017quantum, perdomo2018opportunities, cerezo2021variational}. 
\subsection{Qudit unitaries}

This formulation is quite general, and specific details regarding the unitary operation need to be considered, especially in relation to hardware constraints. In this context, we consider a $d$-level system with a fixed (time-independent) Hamiltonian $H_0$, which has eigenvalues $|l\rangle$, where $l=0, ..., d-1$. Furthermore, we assume that a series of two-level rotations can be applied. These rotations serve as the fundamental operations achievable using conventional Electronic Paramagnetic Resonance (EPR) techniques. Importantly, these operations are \emph{universal} for a single qudit, enabling the generation of any wave function within the $d$-dimensional Hilbert space. In the interaction picture, utilizing a simple monochromatic microwave pulse to coherently control the system, these rotations are given by \cite{castro2022optimal, chiesa2023blueprint}
\begin{equation}
\label{eq:GivensRot}
    \hat{R}_{k,l}(\theta, \phi) = \exp\left(-i \frac{\theta}{2} G_{k,l}(\phi) \right) \; .
\end{equation}
Here, 
\begin{equation}\label{eq:givenspref}
G_{k,l}(\phi) = \left(\cos\phi\cdot \hat{\sigma}_{k,l}^x - \sin\phi\cdot \hat{\sigma}^y_{k,l}\right) \; ,
\end{equation}
where $\hat{\sigma}_{k,l}$ denotes the Pauli matrices derived from the eigenvalues $|k\rangle$ and $|l\rangle$, \emph{i.e.}  ${\sigma}_{k,l}^x = |k \rangle \langle l | + {\rm 
 h.c.}$ and $\sigma_{k,l}^y = i |k \rangle \langle l | + {\rm 
 h.c.}$ .

Each rotation between levels $k$ and $l$ in our system is characterized by two parameters to be tuned: $\theta$ and $\phi$. Consequently, a layer of rotations involving all accessible levels within our qudit will encompass $2(d-1)$ parameters, assuming ladder-like coupling between adjacent levels. It is important to note that these two parameters do not carry the same significance in the resulting state they produce. One parameter, $\theta$, governs the population transfer between the two states, while the other parameter, $\phi$, controls the relative phase they acquire. This distinction becomes relevant when encoding the data, as we will explore later.
The complete ansatz entails applying $L$ layers of these rotations, where the rotation angles store both the classical data, $\vec x$, and the parameters to be optimized, $\vec \varphi$. Therefore, the most general qudit unitaries that we will consider through this work are,
\begin{equation}\label{eq:ansatz}
    \hat{U}(\vec{x}; \vec{\varphi}) = \prod_l^L \left(\prod_i^{d-1} \hat{R}_{i, i+1}^l\left(\theta_i(\vec{x}, \vec{\varphi}), \phi_i(\vec{x}, \vec{\varphi})\right)\right) \; .
\end{equation}
We emphasize that  the rotation angles $\theta_i$ and $\phi_i$ are general functions of both $\vec x$ and $\vec \varphi$.  This is discussed in the next subsection and in the rest of the manuscript the explicit dependence of these angles is dropped to ease the notation.

\subsection{Metric learning}\label{subsec:training}

In our study, we focus on metric learning, specifically its quantum version \cite{lloyd2020quantum, nghiem2021unified}.
Metric learning is well-suited for classifiers, particularly in image classification and recognition.
This type of learning aims to map the data points onto a feature space equipped with a metric such   the distance between points belonging to the same class are minimized while ones  belonging to different classes are maximally separated.

In the quantum realm, 
the unitary $\hat{U}(\vec{x}; \vec{\varphi})$ in Eq. \eqref{Uxphi} maximizes the distances between states belonging to different classes while minimizing them for states within the same class.
In this work, we will consider that the data corresponds to $K$ classes. The algorithm aims to learn how to classify the input vector into one of these classes by minimizing a certain loss function, $\mathcal{L}$, with respect to the parameters $\vec \varphi$, which are the same for every point $\vec x$ in the training ensemble.
The literature \cite{nghiem2021unified} discusses two main approaches: implicit and explicit algorithms.

\subsubsection{Implicit approach}

In the implicit approach, the training data defines $K$ ensembles, one for each class. Using the encoded input vectors as given by Eq. \eqref{Uxphi}, $K$ ensembles  can be defined as follows:
\begin{equation}\label{eq:datapoint_ensemble}
  \rho_k := \frac{1}{N_k} \sum_{i_k}|\psi(\vec x_{i_k}; \vec \varphi)\rangle\langle\psi(\vec x_{i_k}; \vec \varphi)| \; .
\end{equation}
Here, $N_k$ represents the number of training points belonging to class $k$. 
\\
Besides, the loss function to be minimized reads \cite{lloyd2020quantum, nghiem2021unified}
\begin{align}
\label{eq:loss_implicit}
    \mathcal{L}_{\rm{I}} = 1 - \frac{1}{K}\sum_k \, {\rm tr} (\rho_k^2)
     + \frac{2}{K}\sum_{k<l} \, {\rm tr} (\rho_k\rho_l)
\end{align}

We notice that the different terms of the loss function can be  rewritten as,
\begin{equation}\label{eq:trace_terms}
    {\rm tr} (\rho_k\rho_l) = \frac{1}{N_k N_l}\sum_{i_k, i_l}^{N_k, N_l}|\langle 0| U^\dagger(\vec x_{i_k}; \vec \varphi)U(\vec x_{i_l}; \vec \varphi)|0\rangle|^2\ .
\end{equation}
A possible algorithm is to prepare the states $|\psi_{i_k,i_l} \rangle \equiv U^\dagger(\vec{x}_{i_k}; \vec{\varphi}) U(\vec{x}_{i_l}) |0\rangle$, measuring it in the computational basis and extracting its frequency $p_0$. The complexity scales as $O(\epsilon^{-2})$. Here, the largest sampling error is approximately $O(N_{\rm shots}^{-1/2})$, where $N_{\rm shots}$ is the number of shots used to estimate $p_0$ \cite{sancho2022quantum}.

Thus, this approach aims to maximize the purity of the ensemble of states belonging to the same class while moving them away from ensembles of other classes by minimizing the trace of the crossed ensembles.

\subsubsection{Explicit approach}
\label{subsub:explicit}

In the explicit approach, the algorithm utilizes predefined reference states (referred to as centers) one for each class, denoted as $\{ | \psi_k^{\rm R} \rangle \} _{k=1}^K$. Introducing the density matrix,
\begin{equation}
\label{sigmak}
    \sigma_k := | \psi_k^{\rm R} \rangle
    \langle  \psi_k^{\rm R} |\ ,
\end{equation}
the loss function in this case is given by
\begin{equation}\label{eq:loss_explicit}
    \mathcal{L}_{\rm{E}} = 1 - \frac{1}{K}\sum_k \, {\rm tr} (\rho_k\sigma_k)\ .
\end{equation}
Here, the algorithm aims to minimize the distance between the training data and the reference states. It is worth noting that the last term in the above equation is nothing but the fidelity, ${\rm tr} (\rho_k\sigma_k) = F (\rho_k, \sigma_k) :=  \langle \psi_k^{\rm R} |\rho_k | \psi_k^{\rm R} \rangle$, which can again be expressed in a similar fashion as Eq. \eqref{eq:trace_terms}. Using the definition given in Eq. \eqref{sigmak} for $\sigma_k$, ${\rm tr} (\rho_k\sigma_k)$ reads as
\begin{equation}
    {\rm tr} (\rho_k\sigma_k) = \frac{1}{N_k}\sum_{i_k}^{N_k}|\langle 0|U^\dagger(\vec x_{i_k}; \vec \varphi)U_k^R| 0\rangle|^2\ ,
\end{equation}
being $U_k^R$ the unitary operation needed to generate the reference state $|\psi_k^R\rangle$, $|\psi_k^R\rangle = U_k^R|0\rangle$. Pretty much as in \eqref{eq:trace_terms}, the algorithm consists of computing the state $|\psi_{i_k}\rangle \equiv U^\dagger(\vec{x}_{i_k}; \vec{\varphi})U_k^R |0\rangle$ and then computing the frequency $p_0$ in the computational basis.

The performance of the algorithm depends on the selection of the centers. Ideally, they should be as separate as possible. In terms of quantum states, orthogonal states are maximally separated. Therefore, it is natural to choose $\langle \psi_k^{\rm R} | \psi_{k^\prime}^{\rm R} \rangle = \delta_{k, k^\prime}$. However, it is possible that $d < K$. In such cases, we argue below that maximally orthogonal states (MOS) can be employed.

In both approaches, the strategy to follow in terms of optimisation will be the same: minimise the corresponding loss function, $\mathcal{L}$, with respect to the control parameters $\vec\varphi$. In the implicit approach, this will imply to maximize the purity of each ensemble, defined in Eq. \eqref{eq:datapoint_ensemble}, as well as minimize the overlap between states corresponding to data points from different classes. In the explicit approach, the optimization will seek to maximize the overlap between the states generated by each data point and its corresponding reference state. Both approaches, however, can be related through geometrical arguments, as we discuss below.

\subsubsection{Geometric interpretation}\label{subsec:GeometricInterpretation}

In the next section, we are going to test and compare both approaches in actual data sets.
Before that, we find it convenient to provide a geometric way of relating them.

As is clear from the above definitions, the loss functions, \(\mathcal{L}\), we want to minimise are related to distances between states in Hilbert space. There are several metrics one can define, see Chapter 9 in Ref. \cite{wilde2013quantum}. Here, we focus on two: the trace distance and the Hilbert-Schmidt distance.

For two general states, \(\rho\) and \(\sigma\), the trace distance, \(D_T(\rho, \sigma)\), is defined as \(D_T(\rho, \sigma) = \frac{1}{2}\lVert\rho - \sigma\rVert_1\), where \(\lVert A \rVert_1 = \text{Tr}\sqrt{A^\dagger A}\) is the trace norm, while the Hilbert-Schmidt distance, \(D_{\text{HS}}\), is defined as \(D_{\text{HS}} = \text{Tr}\left[(\rho - \sigma)^2\right] = \lVert\rho - \sigma\rVert_2^2\), where \(\lVert A \rVert_2 = \sqrt{\text{Tr}\left(A^\dagger A\right)}\) is the Frobenius norm.

Now, we can re-write Eq. \eqref{eq:loss_implicit} for the particular case of \(K = 2\), i.e., a binary classification problem, as
\begin{equation}
    \label{eq:implicit_HS}
    \mathcal{L}_{\text{I}} = 1 - \frac{1}{2}D_{\text{HS}}(\rho_1, \rho_2).
\end{equation}

Theorem \ref{theo} shows that, for binary classification tasks, the implicit and explicit approaches defined in Eq. \eqref{eq:implicit_HS} and Eq. \eqref{eq:loss_explicit}, respectively, are directly related.

\begin{theorem}
\label{theo}
In a binary classification task, the minimizers $\rho_{k = 1, 2}$ of the explicit loss function $\mathcal{L}_{\rm{E}}$ \eqref{eq:loss_explicit} with reference states $\sigma_{k=1,2}$ \eqref{sigmak} are also the minimizers of the implicit loss function $\mathcal{L}_{\rm{I}}$ \eqref{eq:loss_implicit} if and only if the trace distance between the reference states,  $D_T(\sigma_1, \sigma_2)$, is maximal.
\end{theorem}
\begin{proof}

As discussed in section \ref{subsub:explicit}, we can express the explicit loss function as $\mathcal{L}_{\rm E} = 1 - \frac{1}{2}\sum_k F (\rho_k, \sigma_k)$.
The minimizers, $\rho_{k=1, 2}$ of $\mathcal{L}_{\rm E}$ will satisfy $F(\rho_k, \sigma_k) \geq (1 - \epsilon)$ with $\epsilon$ some arbitrarily small positive number.
This, in turn, implies that $\rho_k$ and $\sigma_k$ satisfy $||\rho_k - \sigma_k||_1\leq 2\sqrt{\epsilon}$, i.e. that $\sigma_k$ and $\rho_k$ are $\sqrt{\epsilon}$-close in trace distance \cite{wilde2013quantum}.
On the other hand, minimizing the implicit loss function, $\mathcal{L}_{\rm{I}}$, is equivalent to maximizing the Hilbert-Schmidt distance $D_{\rm HS}(\rho_1, \rho_2)$. The trace distance, $D_T(\rho_1, \rho_2)$, and thus, $|| \rho_1 - \rho_2 ||_1$ is bounded above and below by multiples of the Hilbert-Schmidt distance, $D_{\rm HS}(\rho_1, \rho_2)$ \cite{coles2019strong}. Thus, maximizing $D_{\rm HS}(\rho_1, \rho_2)$ is equivalent to maximizing $|| \rho_1 - \rho_2 ||_1$. Now, we can relate $|| \rho_1 - \rho_2 ||_1$ to $||\rho_k - \sigma_k||_1$ through the following relation:
\begin{equation}
\begin{aligned}
    ||\rho_k - \rho_{k'}||_1 {} &\leq ||\rho_k - \sigma_k||_1 + \\
    & + ||\rho_{k'} - \sigma_{k'}||_1 + ||\sigma_{k} - \sigma_{k'}||_1 \,.
\end{aligned}
\end{equation}
Then, the minimizers, $\rho_{k=1,2}$, of $\mathcal{L}_{\rm E}$ and thus maximizers of $\mathcal{F}(\rho_k, \sigma_k)$ obey:
\begin{equation}
    ||\rho_k - \rho_{k'}||_1\leq 4\sqrt{\epsilon} + ||\sigma_{k} - \sigma_{k'}||_1 \,.
\end{equation}
Therefore, in order for them to be maximizers of $||\rho_k - \rho_{k'}||_1$ and consequently minimizers of $\mathcal{L}_{\rm{I}}$, the trace distance between centers $||\sigma_{k} - \sigma_{k'}||_1 / 2$ has to be maximal.
\end{proof}

This theorem justifies choosing as centers states that maximize their mutual distance.
These form what we call a set of maximally orthogonal states (MOS). 
\begin{definition}\label{def:mos}
Let $\Phi = \{|\psi_k\rangle \}_{k=1}^K$ be a set of $K$ pure states such that $|\psi_k\rangle \in \mathcal{H}^d$, being $\mathcal{H}^d$ the Hilbert space of the qudit.
We say that $\bar\Phi$ forms a MOS-set iff $\bar\Phi = \argmin_{\{\Phi\}} E_W$ with $E_W$ an energy defined as
\begin{equation}
    E_W(\Phi) =  \sum_{i\neq j} W\left( |\langle \psi_i | \psi_j \rangle|\right) \ ,
\label{eq:energyMOS}
\end{equation}
where $W$, the weighting function, is an arbitrary real-valued and increasing function well-defined in the interval $[0, 1]$.
\end{definition}
With this definition we can use numerical methods to obtain MOS for any configuration $(d, K)$ of qudit dimension, $d$, and number of states, $K$.
Based on this idea, in order to deal with any dataset and any dimensionality of our physical system, we developed a numerical tool to obtain such MOS.
We leave the details of the tool, as well as a discussion of the different possibilities that arise from the choice of the weighting function, $W$, for appendix \ref{app:MOS} and another work \cite{GApreprint}.


\subsection{Encoding strategies}\label{subsec:encoding}

The main difference between our approach and others is that we do not apply independent ansätze for the enconding  \cite{larose2020robust}. 
In the algorithms discussed here, the model learns a metric, so the encoding itself already maximizes the distances between different classes.
 Therefore, only a single ansatz is considered, where the classical data and the parameters to be optimised are tied together \cite{perez2020data, lloyd2020quantum, nghiem2021unified}.
Therefore, taking into account the unitary transformation in Eq. \eqref{eq:ansatz}, we need to discuss the chosen function to map the input vectors and variational parameters to rotation angles, \emph{i.e} 
\begin{equation}
g: (\vec x, \vec \varphi) \mapsto (\vec \theta, \vec \phi)    
\end{equation}
 This, in turn, will give rise to the embbeding
 \begin{align}
     \nonumber
     \mathcal{E}: & \mathbb{R}^{D_x} \rightarrow \mathcal{H}^d
     \\
     &\vec x \mapsto |\psi(\vec x)\rangle
 \end{align}
through our ansatz, Eq. \eqref{eq:ansatz}.

In this work, we choose a function $g$ similar to those used in classical neural networks. Specifically, our optimization parameters, $\vec \varphi$, are divided into a weight matrix, $\bar \omega$, and a bias vector, $\vec b$, such that $g: \vec x \mapsto \vec x' = \bar\omega\cdot\vec x + \vec b$. This transformed vector, denoted as $\vec x'$, is then used to determine the rotation angles $\vec \theta$ and $\vec \phi$.

Several comments are relevant here. First, the dimension of $\vec x'$, denoted as $D_{x'}$, does not have to be equal to the dimension of $\vec x$, i.e., $D_{x'} \neq D_x$. Even though there are $2(d-1)$ parameters in a layer of our ansatz (as shown in Eq. \eqref{eq:ansatz}), we can handle a higher dimension by employing sublayers \cite{perez2020data}. We simply split the vector $\vec x'$ into tuples of dimension $2(d-1)$. It is important to differentiate between using sublayers to accommodate larger data dimensions and using layers in the ansatz to increase the number of optimized parameters. The former allows us to handle larger data dimensions, while the latter adds non-linearity to the model. Second, the assignment of rotation angles $(\vec \theta, \vec \phi)$ from the transformed data $\vec x'$ is arbitrary and will impact the model's performance, as discussed below.


To assess the effectiveness of different encodings, we can study the decision boundaries drawn by our ansatz
with  a concrete example \cite{larose2020robust}. We take a three-level system (a qutrit) with $d = 3$ for classifying data into three classes ($K = 3$). The reference states are the lowest-energy states from the orthonormal basis of the qutrit: $|k\rangle$ for class $k$ with $k = 0, 1, 2$. Suppose the data dimension is $D_x = 4$. We start with the following transformation: $x'_i = \bar\omega_{ii}\cdot x_i + b_i$. In other words, the matrix $\bar \omega$ is chosen to be diagonal, and the angles are assigned as
$(\theta_i, \phi_i) = (x'_{2i}, x'_{2i+1})$, so $\theta_i = \theta_i(x_{2i})$ and $\phi_i = \phi_i(x_{2i+1})$. Since it is a qutrit, we have two available transitions and, thus, four rotation angles in our ansatz. Applying this transformation to the ground state $|0\rangle$ yields the expression in Eq. \eqref{eq:EncodingExample}:
\begin{equation}\label{eq:EncodingExample}
    |\psi(\vec x)\rangle = \hat U(\vec x, \bar \omega, \vec b)|0\rangle = \sum_k c_k|k\rangle
\end{equation}
where the coefficients $c_k$ are as follows: $c_0 = \cos\frac{\theta_0}{2}$, $c_1 = -i\sin\frac{\theta_0}{2}\cos\frac{\theta_1}{2}e^{i\phi_0}$, and $c_2 = -\sin\frac{\theta_0}{2}\sin\frac{\theta_1}{2}e^{i(\phi_0+\phi_1)}$.
From this expression, we can observe that not all the parameters play a role in the decision boundaries. For example, when computing the overlaps with the reference states $|k\rangle$ as in Eq. \eqref{eq:loss_explicit}, the relative phases introduced by the angles $\phi$ do not have any impact. Furthermore, since each angle in this case depends only on a single component of $\vec x$, our model is trained with only half of the information provided by each data point.

We can address this issue in various ways. The simplest approach is to initialize the qutrit in the state $|+\rangle = 1/\sqrt{K}\sum_k|k\rangle$, so the coefficients $c_k$ in Eq. \eqref{eq:EncodingExample} become a more intricate combination of rotation angles, and the projection onto the reference states no longer cancels out the relative phases. Alternatively, we could change the basis in which the reference states are tailored. For instance, instead of choosing the eigenstates of $\sigma_z$, $\{|0\rangle, |1\rangle\}$, for the orthonormal qubit case, we could use the eigenstates of $\sigma_x$, $\{|+\rangle, |-\rangle\} = \{(|0\rangle + |1\rangle)/\sqrt{2}, (|0\rangle - |1\rangle)/\sqrt{2}\}$. In practice, the outcome would be the same.
 
Another option is to apply a different transformation $g$ such that $g:\mathbb R^{D_x} \mapsto \mathbb R^{2(d-1)}$. 
We can consider $\bar\omega$ as a general $(D_{x'}, D_x)$-matrix with $D_{x'} = 2(d-1)$, ensuring that the output always fits into one sublayer. 
However, each $x'_i$ is now a function of all the components of $\vec x$: $x'_i = \sum_{j}\bar\omega_{ij}\cdot x_j + b_i$.
Again, if the initial state were to be the ground state, $|0\rangle$, we would still face the issue of utilizing only half of the parameters when assigning the angles as before, i.e., $(\theta_i, \phi_i) = (x'_{2i}, x'_{2i+1})$. 
However, initializing in the $|+\rangle$ state should address this problem. In Table \ref{tab:EncodingsG}, we summarize the three main choices for the function $g$ that will be used in section \ref{subsec:EncComp} in our  simulations.

\begin{table}[h!]
\centering
\resizebox{0.75\columnwidth}{!}{%
\begin{tabular}{|ccc|}
\hline
\multicolumn{3}{|c|}{Encodings}                                                                                 \\ \hline
\multicolumn{1}{|c|}{$g$}   & \multicolumn{1}{c|}{Transformation}                              & Initialization \\ \hline
\multicolumn{1}{|c|}{$g_1$} & \multicolumn{1}{c|}{$x_i' = \bar\omega_{ii}\cdot x_i + b_i$}       & $|0\rangle$    \\ \hline
\multicolumn{1}{|c|}{$g_2$} & \multicolumn{1}{c|}{$x_i' =\bar\omega_{ii}\cdot x_i + b_i$}        & $|+\rangle$    \\ \hline
\multicolumn{1}{|c|}{$g_3$} & \multicolumn{1}{c|}{$x_i' = \sum_j\bar\omega_{ij}\cdot x_j + b_i$} & $|+\rangle$    \\ \hline
\end{tabular}%
}
\caption{Different encodings that will be used in Section \ref{subsec:EncComp}.}
\label{tab:EncodingsG}
\end{table}

It is important to note that when multiple layers are involved in the ansatz, either to accommodate data with higher dimensions or to introduce non-linearity through more optimizing parameters, or when the reference states are non-orthogonal, the resulting wave function will generally exhibit sufficient complexity for all parameters to influence the decision boundaries.


Having seen different ways of choosing the function that maps our classical data into the control parameters of our physical system, we can also consider other ways of defining the control itself. For example, an interesting scenario arises when the rotations defined in Eq. \eqref{eq:GivensRot} are not defined on two levels of the orthonormal basis of the qudit, but rather on a series of maximally orthogonal states (MOS).
This becomes relevant when the number of classes exceeds the available levels. Let $\{|\varphi_k\rangle\}_{k=1}^{k=K}$ denote the set of $K$ maximally orthogonal states obtained using our genetic algorithm [Cf. App. \ref{app:MOS}].
We redefine our ansatz as the set of pulses between these ``virtual'' states. Consequently, the Pauli matrices in the $G(\phi)$ term of Eq. \eqref{eq:GivensRot} will now take the form $\tau_{k, l}^{x(y)} = (-i)|\varphi_k\rangle\langle\varphi_l| + \text{h.c.}$, with $k<l$.
With this definition, we gain access to $K - 1$ operations that can be decomposed into a series of pulses between the self-states of the physical system, to be implemented on a real platform \cite{castro2022optimal, chiesa2023blueprint}.
By operating between these maximally orthogonal states, it is reasonable to expect that the generated states belonging to different classes are brought closer to regions that are far from each other, potentially ``saving work'' for the optimizer.

Lastly, we will consider a scenario where $D_x \gg d$, making it impractical to apply successive sublayers to encode the original dimension. In such cases, we will explore various classical dimensionality reduction techniques, such as Principal Component Analysis (PCA) or Convolutional Neural Networks (CNNs), to assess their effectiveness in combination with our model.

\section{Results}\label{sec:results}

After establishing the theoretical foundation of learning with qudits, we proceed to test our algorithms on real-world datasets. We have selected three distinct datasets for this purpose: the \emph{Iris} dataset \cite{ucirepo, fisher1936use}, the \emph{Breast Cancer Wisconsin (Diagnostic)} dataset \cite{ucirepo}, and the \emph{MNIST database of handwritten digits} \cite{deng2012mnist}, as described in section \ref{subsec:datasets}. Each dataset possesses unique characteristics that enable us to explore the efficacy of our tools in diverse scenarios.

In all the simulations that will be presented throughout the manuscript, full access to the wave function of the system is available at all times. This is usually referred to as \emph{state-vector simulations}, i.e. the measurement process does not involve any quantum state tomography type protocol.

In section \ref{subsec:EncComp}, we examine various encoding strategies (ansatz) to determine the approach that yields the best results. Subsequently, in section \ref{subsec:MethodComp}, we employ the selected encoding strategy to compare the performance of different loss functions (implicit or explicit method) in various classification problems.
In section \ref{subsec:HDdata}, we address the specific case of datasets with dimensions significantly larger than that of the qudit, exemplified by the MNIST digits dataset. Here, we study a hybrid algorithm that combines classical dimension reduction techniques with a quantum processing unit (our qudit).
Finally, in section \ref{subsec:hardware_eff}, we discuss the different resources needed to experimentally implement our classifier.

\subsection{Datasets}\label{subsec:datasets}

\begin{itemize}
  \item[$\blacksquare$] \underline{\emph{Iris}} \cite{ucirepo, fisher1936use}
\end{itemize}
Fisher's database of different species of the Iris plant is one of the most (if not the most) used dataset for classification problems. It contains 3 different species (classes) and for each one, 50 data points. Each data point is composed by 4 features which are the length and width of both the sepal and petal of each plant. It is a perfect dataset for an initial benchmark due to the small size of the dataset, the low dimensionaty of the data and the small number of classes. Here, we used a training set of $N_{training} = 30$ data points, 10 for each class, and the rest is left for the test set, $N_{test} = 120$.

\begin{itemize}
  \item[$\blacksquare$] \underline{\emph{Breast Cancer Wisconsin (Diagnostic)}} \cite{ucirepo}
\end{itemize}
In this dataset, there are 569 instances of 10 features belonging to cell nucleus from breast masses. It only has two classes: malignant or benign. This is an example of a dataset that comprises a large enough number of features and instances (and still manageable by a single qudit) and defines a binary classification task. We used $N_{training} = 113$ instances for the training set (around the 20\% of the total) and the rest, $N_{test} = 456$, is used for testing.

\begin{itemize}
  \item[$\blacksquare$] \underline{\emph{MNIST Digits}} \cite{deng2012mnist}
\end{itemize}
This last dataset is also a very well known example of image classification. It comprises 60,000 training images and 10,000 testing images of handwritten digits. It has 10 classes, one per each digit (from 0 to 9). The dataset was provided by the \emph{Python} package \emph{Pytorch} and each gray-scale image is a $28 \times 28$ matrix. It sets the last example as the final test for a task that has a reasonable large number of classes (for a single qudit) and a huge data dimension, which is unfeasible to classify by brute force. Using the model described in section \ref{subsec:HDdata} we used 300 images per digit for training and 600 for the testing. Here, since the problem involves a more sophisticated technique, we further split the training set into $N_{training} = 240$ and $N_{validation} = 60$ as discussed in Section \ref{sec:learningwq}. The validation set will be used along the training procedure to compute both the loss function and the accuracy to ensure that no overfitting is taking place.

\subsection{Testing encoding strategies}\label{subsec:EncComp}

We commence by analyzing the various encoding techniques discussed in Section \ref{subsec:encoding} (refer to Table \ref{tab:EncodingsG}) to determine whether any particular technique outperforms the others.
To gain insights, we leverage the Iris dataset. The comparisons are conducted using the explicit method, where the centers are predetermined as described in Section \ref{subsub:explicit}. The chosen performance metric is the test accuracy, which is defined as follows:
\begin{equation}
    \text{Test Accuracy} = \frac{\text{Number of Correct Predictions}}{N_{\text{test}}}
\end{equation}
Figure \ref{fig:ansatzExplicitComparison} illustrates the investigation conducted using various ansatz designs. The upper plot presents a comparison of the test accuracy's impact when employing  non-orthogonal states  for generating superpositions, as opposed to utilizing states from the system's orthonormal basis.
To achieve this, we consider a qubit with the encoding function $g_1$ from Table \ref{tab:EncodingsG} and a single optimization layer. We select three non-orthogonal states, $\{|\varphi_k\rangle\}_{k=0}^{k=2}$, to accommodate the Iris dataset's dimensionality, $D_x = 4$, within a single sublayer. To evaluate the significance of these states,  we start with a non maximally orthogonal set and gradually increase their orthogonality until obtaining the solution provided by MOS.
In the top plot of Figure \ref{fig:ansatzExplicitComparison}, these points correspond to the circles. As anticipated, greater orthogonality among these three states leads to better coverage of the Hilbert space with our pulses, resulting in increased expressibility and test accuracy.

We compare this kind of ansatz with pulsing over the original set of orthonormal states and adding sublayers to accomodate the complete data dimensionality. This result is marked with a black square in the figure and proves to be more efficient for this particular task.

Furthermore, we compare these findings with the performance of a qutrit ($d = 3$) under identical circumstances. The qutrit serves as an ideal unit of information for this problem. It has two  transitions, enabling the embedding of complete data dimensionality within our ansatz, without necessitating sublayers. 
Surprisingly, the performance of the qutrit falls short compared to the performance achieved using a qubit with sublayers. As we will delve into, this is due to the choice of the encoding function, denoted as $g_1$, which, as demonstrated in Section \ref{subsec:encoding}, fails to fully exploit the information provided by each data point.

\begin{figure}[h!]
    \centering
    \includegraphics[width = \columnwidth]{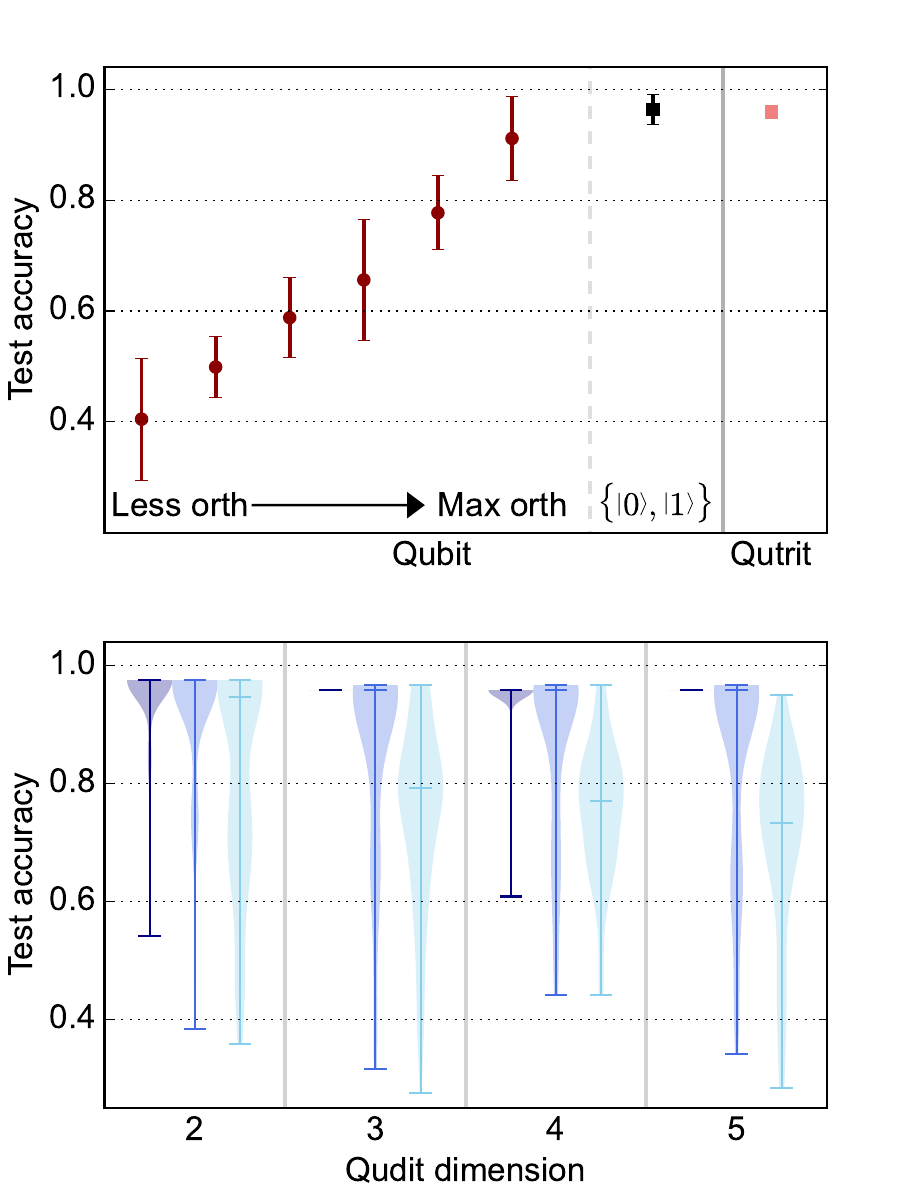}
    \caption{Comparison between ansätze in the explicit framework. Top: Comparison between different ansatz structures (design of the basic pulses). For the red rounded points, a non-orthogonal set of states is used as basis. The black square point represents the results when the original orthonormal basis is used and the last salmon square point marks the result for a qutrit (also with the orthonormal basis). Bottom: Comparison between different embedding functions of the data (how the original data is transformed into the parameters of the pulses). From darker to lighter colors: $g_1$, $g_2$ and $g_3$ from Table \ref{tab:EncodingsG}.
    }
    \label{fig:ansatzExplicitComparison}
\end{figure}

Having determined that the ansatz involving the states of the orthonormal basis of the system (the original definition in Eq. \eqref{eq:GivensRot}) is the most efficient in terms of the type of operations between levels, we now investigate, using this ansatz, which encoding function yields the most beneficial impact on the training process. To accomplish this, we compare the three embeddings, $g_i$ with $i=1,2,3$, presented in Table \ref{tab:EncodingsG}, for different qudit dimensions $d$. Once again, we utilize the Iris dataset, employing only one optimization layer and the explicit method.

The lower plot presents violin plots illustrating the distribution of test accuracies acquired from $100$ distinct and independent initial conditions in the training process. For the darkest color ($g_1$), it is evident that there is negligible statistical variance when $d > 2$. With the sole exception of a single point at $d = 4$, the test accuracy consistently remains at $95.83\%$.
As detailed in Section \ref{subsec:encoding}, this encoding strategy is suboptimal because it overlooks half of the original data, $\vec x$, treating it as relative phases in the generated states. These phases have no influence when projecting onto the reference states of the orthonormal basis. Although it appears to yield good results for this specific dataset, it may not do so in other scenarios where unutilized features play a crucial role.
Hence, it is more appropriate to consider the other two alternatives, namely, $g_2$ and $g_3$ (represented by medium and light colors in the plot, respectively). It is noteworthy that, despite employing a larger number of parameters in the optimization process, the median performance of $g_3$ falls short of that achieved with $g_2$.
Consequently, based on these findings, we deduce that the most suitable option among these alternatives is to encode our data using $g_2: x_i' \mapsto \bar\omega_{ii}\cdot x_i + b_i$, initializing our qudit in the $|+\rangle$ state. Consequently, we adopt this configuration for all subsequent simulations.

\subsection{Quantum metric learning}\label{subsec:MethodComp}

We now proceed to compare the performance of both the explicit and implicit methods discussed in section \ref{subsec:training}.
As concluded at the end of section \ref{subsec:EncComp}, the encoding used for the subsequent simulations with the explicit method is $g_2$ from Table \ref{tab:EncodingsG}.
However, for the implicit method, we adhere to using $g_1$. In other words, we initialize the state to $|0\rangle$. The rationale behind this decision is that, since there are no fixed centers in this case, there is no issue with unused relative phases. Thus, for the sake of simplicity, we find it more convenient to initialize the system in the ground state. It is worth noting that simulations were also conducted by initializing the system in the $|+\rangle$ state, yielding similar results.

\begin{figure}[h!]
    \centering
    \includegraphics[width = \columnwidth]{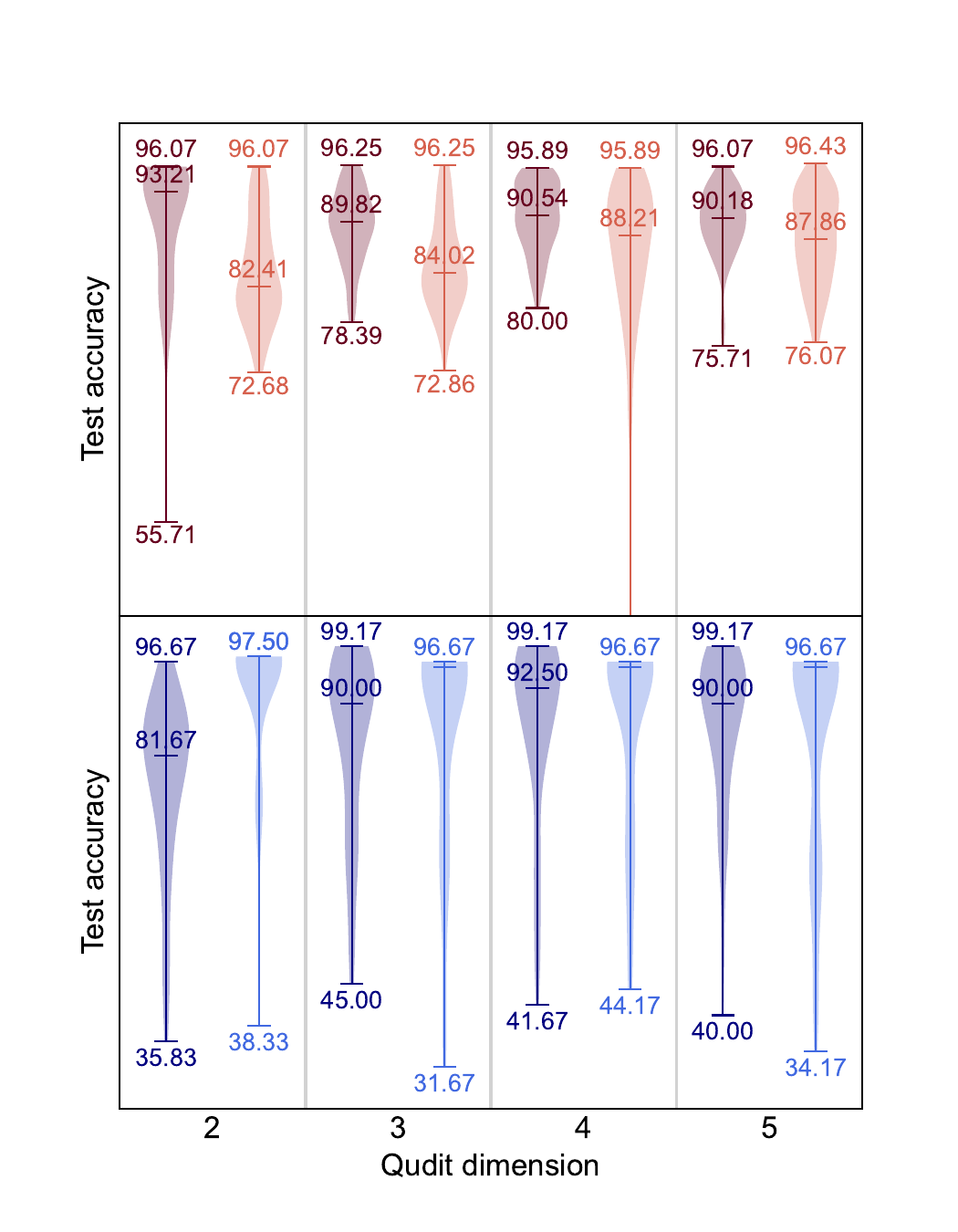}
    \caption{Comparison between frameworks for different datasets. Top: Results for the Breast cancer Wisconsin data set. Bottom: Results for the Iris data set. Darker colors represent the results for the implicit method and lighter colors represent the explicit method.}
    \label{fig:MethodResults}
\end{figure}

Figure \ref{fig:MethodResults} illustrates the results obtained for the Breast Cancer dataset (top plot) and the Iris dataset (bottom plot). The darker-colored violin plot represents the outcomes obtained with the implicit method, while the lighter color represents the explicit method.
Each distribution in the plot is derived from 100 different and independent initial conditions, and a single optimization layer is employed.
It is evident that, although the implicit method generally appears to achieve better results, both methods demonstrate a good performance. In fact, these results are comparable to those achievable with the best classical classifiers for these datasets \cite{ucirepo}. A comparison between classical models and the findings of this study is presented in Figure \ref{fig:modelcomp}.

\begin{figure}[h]
    \centering
    \includegraphics[width = \columnwidth]{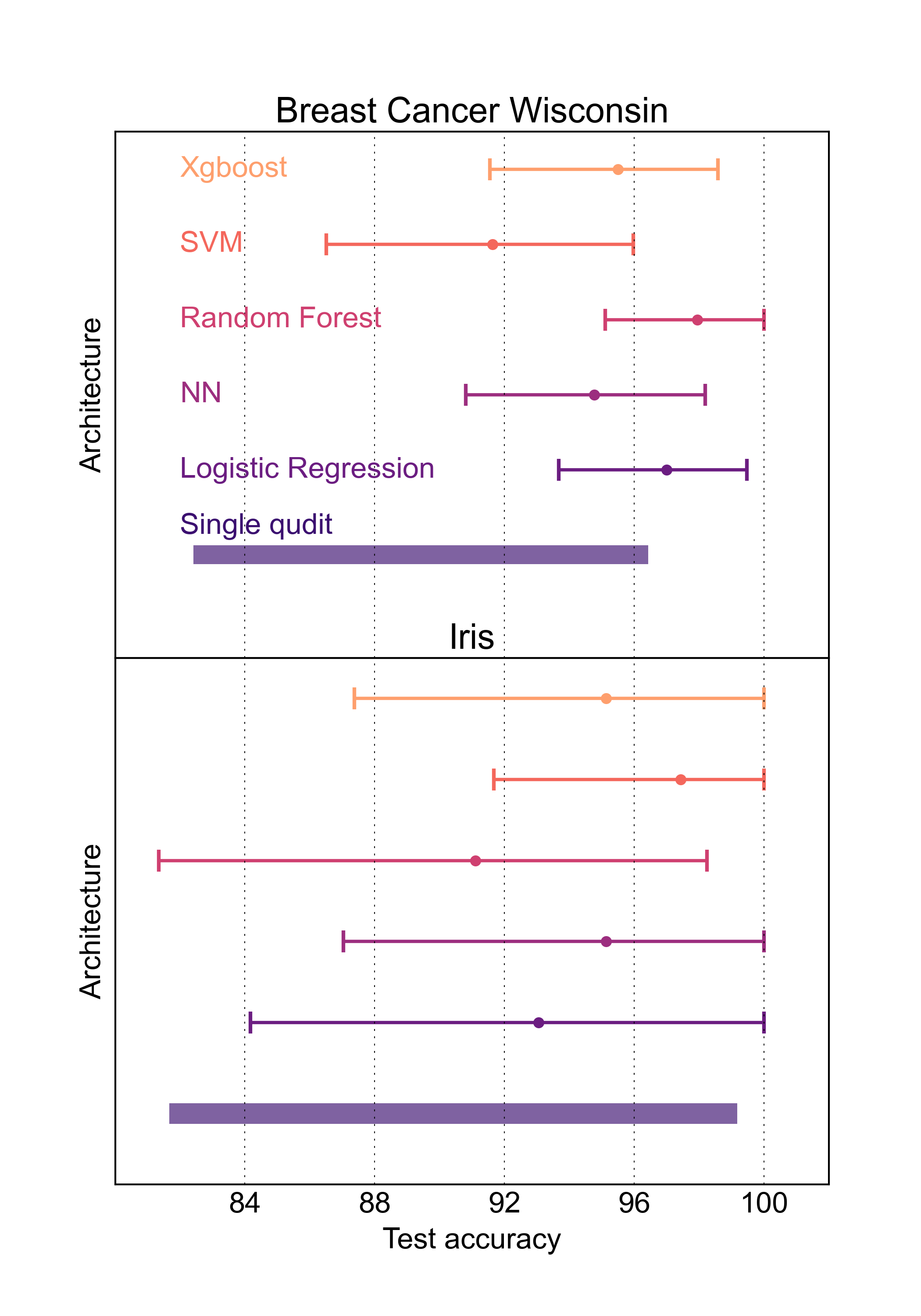}
    \caption{Performance comparison between the best classical models and our best results. The results for the classical models were extracted from \cite{ucirepo} and the bars for the single qudit results cover the distribution of medians and maximum test accuracies obtained for all qudit dimensions in Figure \ref{fig:MethodResults}.}
    \label{fig:modelcomp}
\end{figure}

It is worth emphasizing that the dimensionality of the qudit does not appear to be a critical factor, except in the case of the Iris dataset where $d > 2$ seems to yield slightly better results compared to $d = 2$. This suggests that in datasets where the number of classes does not allow for exploiting the capacity of the MOS ($K \approx d$), the size of the Hilbert space does not play a significant role.

Furthermore, we can understand that the implicit method is capable of finding slightly superior solutions compared to the explicit method because it has the flexibility to discover the centers that best fit the data structure. This characteristic becomes evident when we evaluate the clustering achieved by the implicit method through the minimization of Eq. \eqref{eq:loss_implicit} \cite{implicitMOS}. Figure \ref{fig:EmbeddingComparison} provides a comparison of the mapping results obtained by both methods using the test set of the Iris dataset, considering a single qubit and a single layer. The reference states are indicated by arrows, and the color plot represents the overlap between them.
While the states obtained using the explicit method are perfectly homogeneous, we observe that the states generated by the implicit method, although very similar, are not completely homogeneous. This, along with the difference in the definition of each loss function, is reflected in the distribution of points on the Bloch sphere. In the explicit case, the points are almost perfectly aligned on the circumference connecting the three reference states, whereas in the implicit case, there is a greater dispersion on the surface.

\begin{figure}[h!]
    \centering
    \includegraphics[width = \columnwidth]{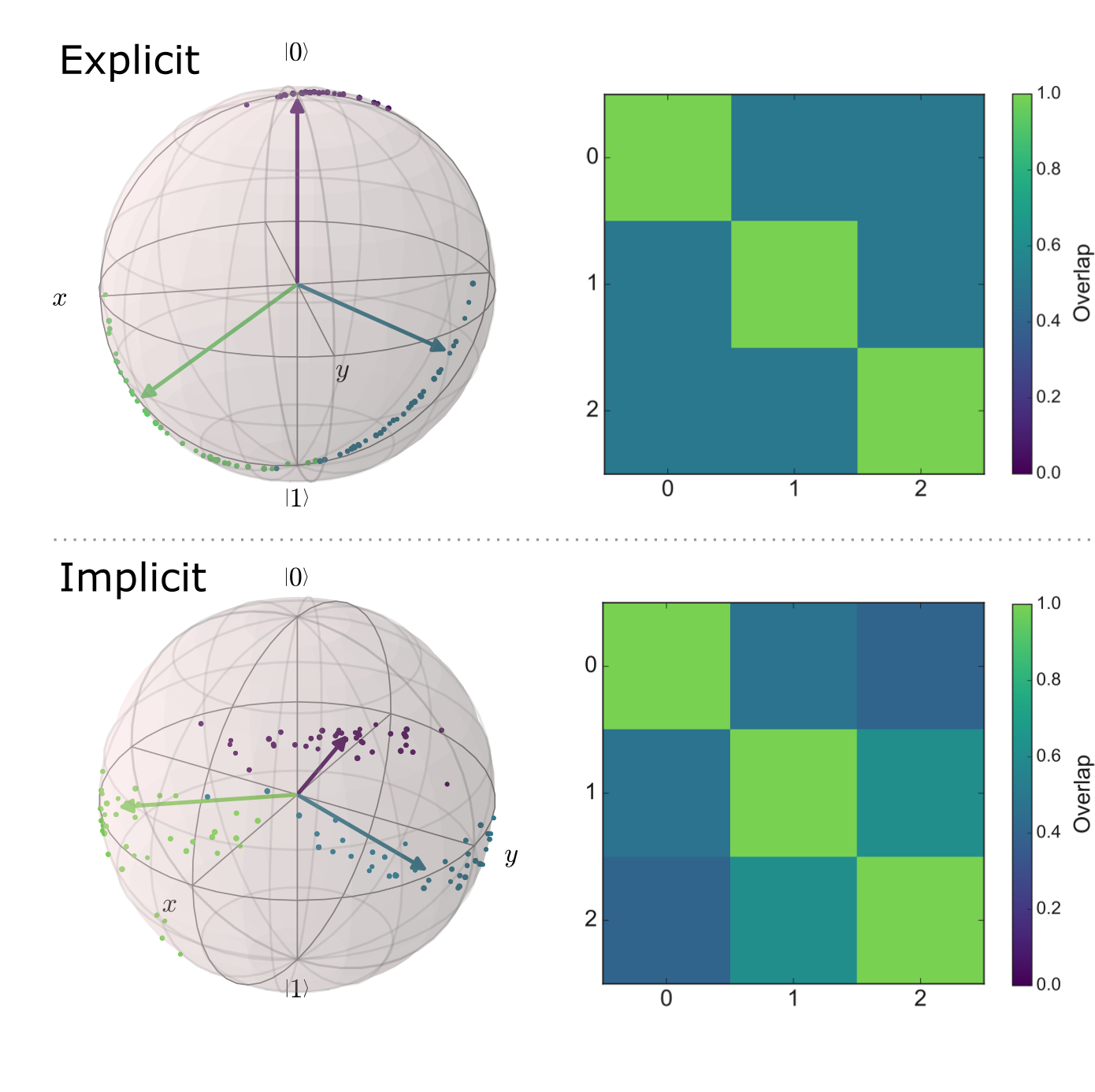}
    \caption{Different mappings produced by each method in the Iris data set. On the left side of the figure we represent each wavefunction produced by the circuit after training. The color of each point represents the class at which it belongs. Same for the arrows (reference states). On the right side we plot the overlap between these reference states. Although not perfectly, the implicit method manages to find states that are practically homogeneous and maximally orthogonal in this case.}
    \label{fig:EmbeddingComparison}
\end{figure}

\subsection{High-dimensional data}\label{subsec:HDdata}

\begin{figure*}[t!]
    \centering
    \includegraphics[width = 0.85\textwidth]{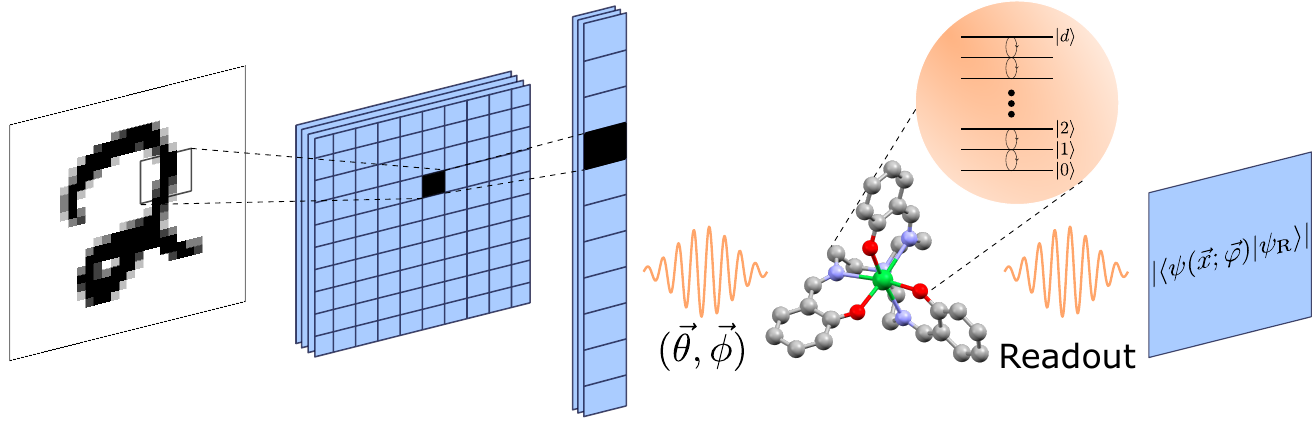}
    \caption{Hybrid Convolutional Neural Network (HCNN). In this model, we propose the use of Convolutional Neural Networks to act as a dimension reducer technique that preserves the spatial correlations present in an image. We exemplify the $d$-level system with its possible implementation in a molecular qudit \cite{chiesa2023blueprint}.}
    \label{fig:hybridcnn}
\end{figure*}

As a final case study, we selected the MNIST dataset to demonstrate the challenge of dealing with high-dimensional data compared to the dimension of the qudit. Applying the data re-uploading technique used in previous cases is impracticable due to the dataset's large dimensionality, $D_x \gg d$.

To address this issue, we need to employ a dimensionality reduction technique that is compatible with the data re-uploading process of our qudit. One commonly used technique is Principal Component Analysis (PCA) \cite{pearson1901liii}. PCA transforms the original data into a new set of uncorrelated variables with lower dimensionality.
These new variables, known as components, capture the directions with the highest variance in the data. The first few components are more relevant than the latter ones, aiming to compress the relevant information from the original dataset.
A challenge with applying PCA in this case is that the images in the dataset need to be ``flattened'' from a two-dimensional format into a one-dimensional vector. This flattening process eliminates the spatial correlations that are crucial for distinguishing between different digits (e.g., distinguishing a 6 from an 8).
In Appendix \ref{app:PCA}, we present a study demonstrating that even with PCA, it is necessary to retain a dimension comparable to that of the original problem to avoid losing valuable information for classification purposes.

Given this challenge, we turn our attention to another technique that has shown remarkable effectiveness in image classification and processing: convolutional neural networks (CNNs) \cite{schmidhuber2015deep}. CNNs consist of two-dimensional processing layers designed to capture spatial correlations present in images. These initial layers extract information from arbitrary matrix blocks rather than individual pixels. This information is then processed by linear layers that assign probabilities to each class. In classical machine learning, a final layer with 10 neurons provides the probabilities for assigning each digit to the image. Our proposed model follows the philosophy of hybrid variational algorithms, combining the optimization capacity of classical computers with the state synthesis capacity of quantum computers. Thus, we design a hybrid convolutional network (HCNN), illustrated in Fig. \ref{fig:hybridcnn}, where the first part consists of classical convolutional and linear layers that preprocess the data and reduce its dimensionality. The processed data is then fed into the quantum layer, which acts as our qudit classifier.

For the design of this hybrid network, we utilized the Python library specialized in neural networks, PyTorch \cite{paszke2017automatic}. Our quantum layer is a customized layer integrated into the global optimization process, and the output dimension of the last classical layer is set to $2(d-1)$. Thus, the CNN compresses the original data to variables of dimension $2(d-1)$, accommodating the qudit dimension, namely
$g: \mathbb{R}^{D_x} \mapsto \mathbb{R}^{2(d-1)}$. We remark that the last layer in the CNN implements the transformation $g_3$ in Table \ref{tab:EncodingsG}. 
We strive to keep the structure of the CNN as simple as possible to avoid the classical part overshadowing the quantum part in the classification process. Consequently, we choose a CNN with two convolutional layers. The first layer has one input channel (as the images are grayscale) and 10 output channels, while the second layer has 10 input channels and 20 output channels, with a kernel size of $(5,5)$. Dropout is applied to mitigate overfitting, and the network concludes with two linear layers: the first with an input dimension of 320 and an output dimension of 50, and the final layer feeds into the quantum layer.

\begin{figure*}[t!]
    \centering
    \includegraphics[width =0.85\textwidth]{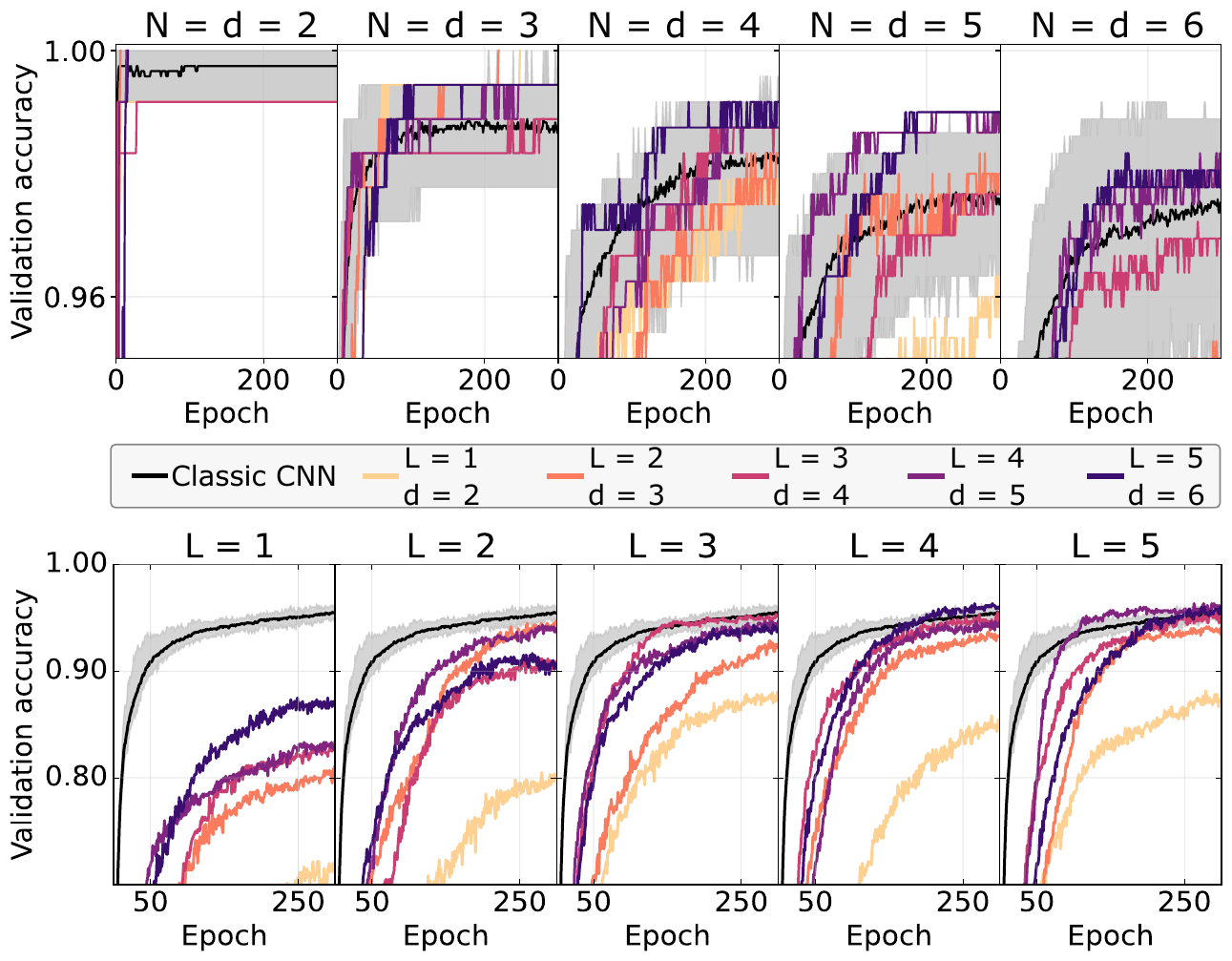}
    \caption{Results obtained for the HCNN model. Top: Results for the partial uses of the data set. We use the same number of digits as levels available in the qudit. Bottom: Results for the full data set (10 digits). Line colors represent the number of layers used (top image) and the number of levels in the qudit (bottom image). In both cases we plot the validation accuracy as it is the one of which we keep track along all the epochs.}
    \label{fig:HCNNresults}
\end{figure*}

Figure \ref{fig:HCNNresults} presents the results of our study. The top plot explores the behavior of the hybrid network when the number of digits to be classified aligns with the number of levels of the qudit, denoted as $d$, as a function of the number of layers, $L$, in the ansatz. The reference states correspond to the orthonormal basis states of the system.

The lower plot illustrates the performance of our classifier when confronted with the complete dataset of 10 digits, as a function of both the qudit dimension and the number of layers in the ansatz. In all scenarios, the hybrid CNN optimizes both its classical and quantum layers simultaneously. This allows us to assess the impact of the quantum layer on the overall classification process. For all these simulations, the explicit method with fixed centers was employed.

The black line, in all scenarios of Fig. \ref{fig:HCNNresults}, represents the accuracy obtained using the exact same classical preprocessing CNN utilized in conjunction with the quantum layer in the hybrid model. However, in this case, the quantum layer is replaced by a linear classical layer with $N$ outputs, where $N$ is the number of digits being classified in each scenario.
The loss function employed is the Cross Entropy function available in the PyTorch package. Additionally, the shaded grey region covering the black line indicates the dispersion in results obtained from 10 different instances of the algorithm, each with different initial conditions for the parameters.
The shaded area represents the maximum and minimum validation accuracy values obtained during the training process for the validation set, while the solid line represents the mean value.
Unfortunately, replicating these statistics for the hybrid model was not feasible due to the computationally intensive training process it requires.

The study reveals that when $N = d$, the hybrid model achieves $100 \%$ accuracy for 2 and 3 digits, surpassing the mean accuracy of the classical network for a sufficient number of layers in all cases. However, when considering the full dataset, the quantum layer appears to act as a bottleneck, resulting in a lower accuracy compared to the classical counterpart. Nevertheless, as more non-linearity (i.e. layers) is incorporated into the quantum ansatz, the hybrid model progressively approaches the performance of the classical results. In Figure \ref{fig:digitsembedding} we further visualise how adding layers (non-linearity) favours the separation between clusters of points on the Bloch sphere for the case of a qubit.

To conclude, table \ref{tab:testaccsMNIST} shows the results of the test accuracy obtained for $N_{test} = 600$ images/digit. We see that although in Figure \ref{fig:HCNNresults} it seemed that the qudit was achieving better results than the classical network alone, in the test data set fails to surpass it for the number of levels and layers that we were able to explore. 

\begin{table}[h!]
\centering
\resizebox{0.75\columnwidth}{!}{%
\begin{tabular}{|cccccc|}
\hline
\multicolumn{6}{|c|}{Test accuracy (\%)}                                                                                                                 \\ \hline
\multicolumn{1}{|c|}{d} & \multicolumn{1}{c|}{L = 1} & \multicolumn{1}{c|}{L = 2} & \multicolumn{1}{c|}{L = 3} & \multicolumn{1}{c|}{L = 4} & L = 5 \\ \hline
\multicolumn{1}{|c|}{2} & \multicolumn{1}{c|}{68.8}  & \multicolumn{1}{c|}{76.8}  & \multicolumn{1}{c|}{84.5}  & \multicolumn{1}{c|}{83.1}  & 85.3  \\ \hline
\multicolumn{1}{|c|}{3} & \multicolumn{1}{c|}{78.1}  & \multicolumn{1}{c|}{91.9}  & \multicolumn{1}{c|}{90.3}  & \multicolumn{1}{c|}{90.6}  & 91.2  \\ \hline
\multicolumn{1}{|c|}{4} & \multicolumn{1}{c|}{80.2}  & \multicolumn{1}{c|}{89.3}  & \multicolumn{1}{c|}{92.1}  & \multicolumn{1}{c|}{92.5}  & 92.8  \\ \hline
\multicolumn{1}{|c|}{5} & \multicolumn{1}{c|}{83.3}  & \multicolumn{1}{c|}{93.6}  & \multicolumn{1}{c|}{94.0}  & \multicolumn{1}{c|}{94.7}  & 94.5  \\ \hline
\multicolumn{1}{|c|}{6} & \multicolumn{1}{c|}{82.7}  & \multicolumn{1}{c|}{89.9}  & \multicolumn{1}{c|}{93.7}  & \multicolumn{1}{c|}{94.1}  & 93.1  \\ \hline
\multicolumn{1}{|c|}{Classic} & \multicolumn{5}{|c|}{$95.26 \pm 0.14$}   \\ \hline
\end{tabular}%
}
\caption{Test accuracy obtained from the best models. Each best model is defined as the one with the highest validation accuracy throughout the training.}
\label{tab:testaccsMNIST}
\end{table}

\begin{figure}[h!]
    \centering
    \includegraphics[width = \columnwidth]{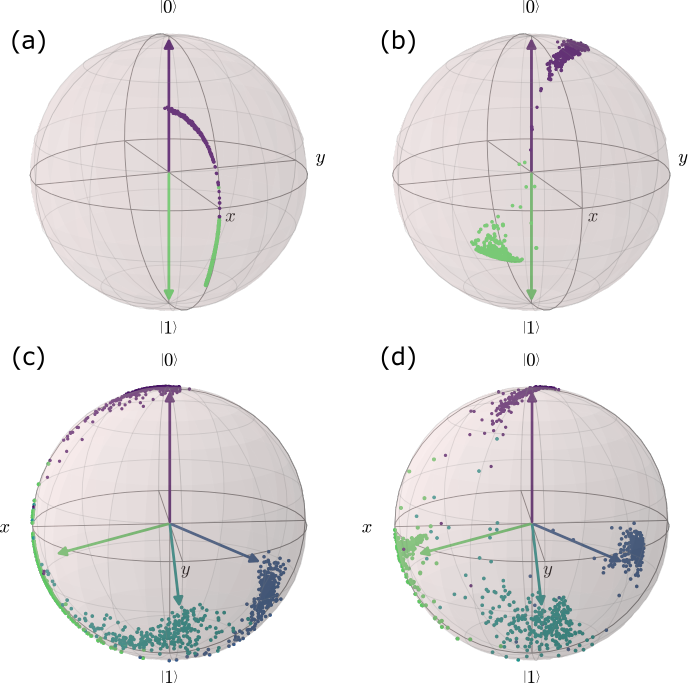}
    \caption{Mapping for the MNIST data set with 2 digits (plots (a) and (b)) and with 4 digits (plots (c) and (d)) using the explicit method. Left side plots are the mappings obtained for $L = 1$ and right side plots the ones obtained for $L = 5$.}
    \label{fig:digitsembedding}
\end{figure}

\subsection{Computational complexity}\label{subsec:hardware_eff}

Throughout different examples we have been discussing the differences in terms of performance between the explicit and implicit methods introduced in section \ref{subsec:training}. Now, we want to discuss about the challenge that each method imposes in a possible experimental implementation.

The implicit method involves a higher training load. The absence of fixed centres means a larger number of terms to evaluate, essentially the cross terms of the type ${\rm Tr}(\rho_i\rho_j)$. These can be evaluated with different techniques like the SWAP test \cite{nielsen2002quantum} or the inversion test \cite{lloyd2020quantum, nghiem2021unified}. In any case,  this is an overhead in terms of processing compared to the explicit case.

Another advantage of the explicit approach is that  no full state tomography is needed in order to evaluate the loss function. Given that our ansatz is universal, i.e. that it is able to generate any wavefunction \cite{castro2022optimal, chiesa2023blueprint}, and that we know the exact coefficients of each reference state, we can relate both to design a unitary operation that implements the desired reference state. Let's take a general reference state 
\begin{equation}
|\psi_R\rangle = \sum_{j = 0}^{d-1} c_j e^{i\beta_j} |j\rangle
\end{equation}
with $c_j, \beta_j \in \mathbb{R}$ and $\beta_0 = 0$. Then, the unitary operation that transforms the ground state $|0\rangle$ into $|\psi_R\rangle$ is
\begin{equation}\label{eq:refUnit}
    |\psi_R \rangle = U_R |0 \rangle= \prod_{k = 0}^{d - 2} \hat R_{k, k+1}(\theta_k, \phi_k)
    |0 \rangle
\end{equation}
with $\phi_k = -\beta_{k+1} - (k+1)\pi/2 - \sum_{l < k}\phi_l$ and $\theta_k = 2\arccos{(c_k/P_k)}$, being $P_k$ a factor which is equal to $1$ for $k= 0$ and $P_k = \prod_{l<k}\sin{(\theta_l/2)}$ otherwise.
To compute the loss function as defined in Eq.\eqref{eq:loss_explicit}, we could just apply $U_R^\dagger$ after the ansatz operation, $U(\vec x; \vec \varphi)$, and just by measuring the population of the ground state we will get the fidelity between the reference and the produced state: 
\begin{equation}\label{eq:GSfid}
|\langle\psi_R|\psi(\vec x; \vec \varphi)\rangle| = |\langle 0|U_R^\dagger U(\vec x; \vec \varphi)|0\rangle| \; .
\end{equation}
This can be related to Eq.\eqref{eq:loss_explicit} through Eq.\eqref{eq:datapoint_ensemble} by noting that $tr\left(\rho_k|\psi_k^R\rangle\langle\psi_k^R|\right) = \frac{1}{N_k}\sum_{i_k}|\langle\psi(\vec x_{i_k}; \vec \varphi)|\psi_k^R\rangle|^2$, so we can reformulate $\mathcal{L}_{\rm E}$ as in Eq.\eqref{eq:loss_explicit_sim}.
\begin{equation}\label{eq:loss_explicit_sim}
    \tilde{\mathcal{L}}_{\rm E} = 1 - \frac{1}{K}\sum_k\frac{1}{N_k}\sum_{i_k}|\langle 0|U_R^\dagger U(\vec x; \vec \varphi)|0\rangle|^2
\end{equation}

Just by measuring the population of the state $|0\rangle$ (for which $N_{shots}$ will be needed depending on the experimental platform and its read-out noise) for each data point, we can obtain the value of the loss function in $N_{shots}\cdot N$ measurements.
Therefore, in the training phase, we will need to run $N_{shots}\cdot N_{training}\cdot N_{eval}$ times the circuit, where $N_{eval}$ is the number of evaluations that the circuit needs to be run in turn by the classical optimizer to achieve convergence.
Typically, $N_{eval} = N_{epochs}\cdot N_{ge}$ where $N_{epochs}$ stands for the number of epochs that the optimizer searches for the solution and $N_{ge}$ is the number of evaluations per epoch that the classical algorithm needs to estimate the gradient.
Finally, in the test phase, we will need to run the circuit up to $N_{shots}\cdot N_{test}\cdot K$ times, as each test point needs to be compared with each of the reference states in order to assign it the most probable class.


\section{Dissipation and decoherence}\label{sec:Dissipation}

Finally, we investigate how our model behaves in the presence of both decay and decoherence.

To achieve this, we formulate a Lindblad-like master equation that realistically models a quantum system. The dynamics is given by \cite{breuer2002theory, Rivas2012},
\begin{equation}\label{eq:mastereq}
    \dot\rho = -i[\mathcal{\tilde{\mathcal{H}}}, \rho] + \sum_{i = 1}^2\mathcal{D}_{T_i}\ .
\end{equation}
The first term on the right-hand side represents the unitary part and it reads,
\begin{equation}
\label{Hinte}
\tilde{\mathcal{H}} = \frac{\Omega_R}{2} \left( e^{i\phi}|j\rangle\langle j+1| + \text{h.c.} \right) \, .    
\end{equation}
Without noise, this term induces the unitary evolution described by Eq.\eqref{eq:GivensRot}.
Physically, this Hamiltonian arises from a driven system
$
\mathcal{H} = \sum_j\epsilon_j|j\rangle\langle j| + \Omega_R\cos(\omega_d t + \phi)\sum_{k,l}(|k\rangle\langle l| + \text{h.c.})\delta_{k, l+1}
$.
Here, $\epsilon_j$ represents the energy of the eigenstate $|j\rangle$, $\Omega_R$ is the strength of the perturbation (Rabi frequency), and the $\delta_{k, l+1}$ term signifies coupling between neighboring levels in the qudit.
Moving to the interaction picture and employing the Rotating Wave Approximation (RWA) when the driving frequency $\omega_d= \omega_{j, j+1} \equiv \epsilon_{j+1} - \epsilon_j$, we arrive at the Hamiltonian \eqref{Hinte}.
The driving frequency, $\omega_d$, determines the addressed transition. The duration of the pulse, $t_p$, corresponds to the rotation angle $\theta$ via the relation $\Omega_R\cdot t_p = \theta$, and the perturbation phase, $\phi$, perfectly matches the phase in the $G(\phi)$ term (see Eq. \eqref{eq:givenspref}).
\\
The second term on the right-hand side of the master equation \eqref{eq:mastereq} accounts for dissipation and decoherence effects. Specifically, the dissipation operators $\mathcal{D}_{T_i}$ are defined as $\mathcal{D}_{T_i} = \frac{1}{T_i}\left[2\mathcal{O}_i\rho\mathcal{O}^\dagger_i - \{\mathcal{O}^\dagger_i\mathcal{O}_i, \rho\}\right]$, where $\mathcal{O}_1 = \sum_j |j\rangle\langle j+1|$ represents decay, and $\mathcal{O}_2 = \sum_j j|j\rangle\langle j|$ represents dephasing.

\begin{algorithm}[H]
  \caption{ Noisy simulations of a single qudit classifier. }
  \label{alg:1}
   \begin{algorithmic}
   \State\textbf{\# Training phase}
   \State Initialize $\vec \varphi$: $\vec \varphi_0$
   \For {$n$ in $N_{eval}$}
       \For {$(\vec{x_i}, y_i)$ in $N_{training}$}
       \State $\rho_0 \gets |+\rangle\langle+|$
       \State $(\vec{t_{p, i}}, \vec \phi_i) \gets\vec{x_i}' \gets g_2(\vec x_i, \vec \varphi_n)$
       \State Define $U(\vec x_i, \vec \varphi_n)$ from $(\vec{t_{p, i}}, \vec \phi_i)$
       \State Define $U_R^\dagger$ from $|\psi^R_{y_i}\rangle$ \Comment{Eq. \eqref{eq:refUnit}}
       \State Solve $\dot\rho_i(t)$ \Comment{Eq.\eqref{eq:mastereq}}
       \State Compute $\mathcal{F}(\rho, |0\rangle)$ \Comment{Eq. \eqref{eq:GSfid}}
       \State Compute $\tilde{\mathcal{L}}_{\rm E}$ \Comment{Eq.\eqref{eq:loss_explicit_sim}}
       \EndFor
   \State Update $\vec \varphi$: $\vec \varphi_{n+1} \gets \vec\varphi_n$
   \EndFor

    \State Store the best parameters $\vec\varphi_{opt}$
    
   \State\textbf{\# Testing phase}
    \For {$\vec{x_i}$ in $N_{test}$}
       \For {$k$ in $K$}
       \State $\rho_0 \gets |+\rangle\langle+|$
       \State $(\vec{t_{p, i}}, \vec \phi_i) \gets\vec{x_i}' \gets g_2(\vec x_i, \vec \varphi_{opt})$
       \State Define $U(\vec x_i, \vec \varphi_{opt})$ from $(\vec{t_{p, i}}, \vec \phi_i)$
       \State Define $U_R^\dagger$ from $|\psi_k^R\rangle$ \Comment{Eq. \eqref{eq:refUnit}}
       \State Solve $\dot\rho_i(t)$ \Comment{Eq.\eqref{eq:mastereq}}
       \State Compute and store $\mathcal{F}(\rho, |0\rangle)$ \Comment{Eq. \eqref{eq:GSfid}}
       \EndFor
   \State Assign class to the point according to $max\{\mathcal{F}(\rho, |\psi_k^R\rangle)\}$
   \EndFor

   \end{algorithmic}
\end{algorithm}

Each operation of our ansatz corresponds to the dynamics governed by Eq. \eqref{eq:mastereq}.
The parameters to be optimized are encoded in the time length of the evolution, $t_p$, as well as in the driving phase, $\phi$. These, in turn, are functions of the original data points, $\vec x$, which are re-scaled according to the functions described in Section \ref{subsec:encoding}.
Once the state of the system has been generated, $\rho(\vec x; \vec \varphi)$, we perform the measurement scheme described in Section \ref{subsec:hardware_eff} so that, in the end, we compute the fidelity between the final state $\rho(t)$ and the ground state, $\mathcal{F}(\rho(t), |0\rangle)$. That is, for each point belonging to $N_{training}$, once the sequence of pulses corresponding to the ansatz, $U(\vec x, \vec \varphi)$, has been performed, we apply the sequence corresponding to the unitary that implements the reference state assigned to that point, $U_R^\dagger$. By doing so, the result from measuring the population of the $|0\rangle$ gives us the value for computing $\tilde{\mathcal{L}}_{\rm E}$ in Eq. \eqref{eq:loss_explicit_sim}.
In the training phase we do this only once per data point, since there is only one reference state to compute the fidelity with. For the testing phase, we measure the overlap between the generated state and all the reference states instead, to be able to decide with which one has the most overlap.
Both the fidelity estimation and the preparation of the initial state are assumed to be ideal. That is, the fidelity is computed directly from the final density matrix (we do not perform neither quantum state tomography nor a stochastic measurement procedure) and the initial state is defined to be a pure state, $\rho(t = t_0) = |+\rangle\langle +|$.
An scheme of this procedure is sketched in the pseudo-code of Algorithm \ref{alg:1}.

To evaluate the performance of the aforementioned model, we employ the Iris dataset.
Figure \ref{fig:noisesims} depicts the results for both a qubit and a qutrit as a function of the system's decoherence time, denoted as $T_2$, which constitutes the main limitation in practical implementations.  We choose parameter values in accordance with typical experimental settings \cite{cao2023emulating, ringbauer2022universal, low2023control, rollano2022high, fischer2022towards}: $\omega_{ij}/2\pi \sim 3$ GHz, $\Omega_R/2\pi = 10$ MHz, $T_1 = 100$ ms, and $T_2 \in$ [100 ns, 100 $\mu$s].
Furthermore, we explore the impact of adding multiple layers to the model and compare the medians of the resulting test accuracies.
For classical optimization, we employ the SPSA optimizer \cite{spall1998overview}, which utilizes only two circuit evaluations per iteration ($N_{ge} = 2$), regardless of the dimensionality of the parameter space being optimized. The maximum number of iterations is set at 30 ($N_{epochs} = 30$, yielding $N_{eval} = 60$). We conduct 50 distinct ``experiments'' (runs) to gather statistics. The initial conditions for each ``experiment'' are derived from the best parameters obtained in the preceding run if there was an enhancement in test accuracy compared to its predecessor; otherwise, they are randomly reinitialized.

\begin{table}[h!]
\centering
\resizebox{0.75\columnwidth}{!}{
\begin{tabular}{|c|c|c|c|}
\hline
System               & Layers & Max Test Accuracy & $T_2$      \\ \hline
\multirow{2}{*}{Qubit}  & L = 1  & 0.983          & 1.5 $\mu$s \\ \cline{2-4} 
                        & L = 2  & 0.992          & 400 ns     \\ \hline
\multirow{2}{*}{Qutrit} & L = 1  & 0.975          & 100 $\mu$s \\ \cline{2-4} 
                        & L = 2  & 0.992          & 6 $\mu$s   \\ \hline
\end{tabular}%
}
\caption{Test accuracies for the Iris dataset using the noisy model. The maximum test accuracy values shown correspond to the minimum $T_2$ value for each respective row.}
\label{tab:noisyresults}
\end{table}

The analysis reveals that by extending the decoherence time, the results quickly converge to those attained in ideal simulations. Notably, with $T_2$ values of 400 ns for the qubit and 1.5 $\mu$s for the qutrit, we achieve maximum accuracies of 0.975 and 0.941, respectively.
Additional numerical values are tabulated in Table \ref{tab:noisyresults}.
Hence, \emph{ the performance of the qudit remains robust even when subjected to realistic noise levels.
}

We can understand the reason behind the  failure  of the model when the noise is increased ($\Omega_R\cdot T_2 \approx 1$). In this scenario, the noise terms of the Quantum Master Equation, the second term in the right hand side of \eqref{eq:mastereq},  become predominant. These dissipators drive $\varrho$ towards its steady state, which, in our case, is $|0\rangle$. As a result, the test accuracy remains around $1/3$. Moving to longer $T_2$ values, the unitary driven part becomes capable of reaching the other reference states (prior to their relaxation), thereby enhancing the accuracy.

\begin{figure}[h!]
    \centering
    \includegraphics[width = \columnwidth]{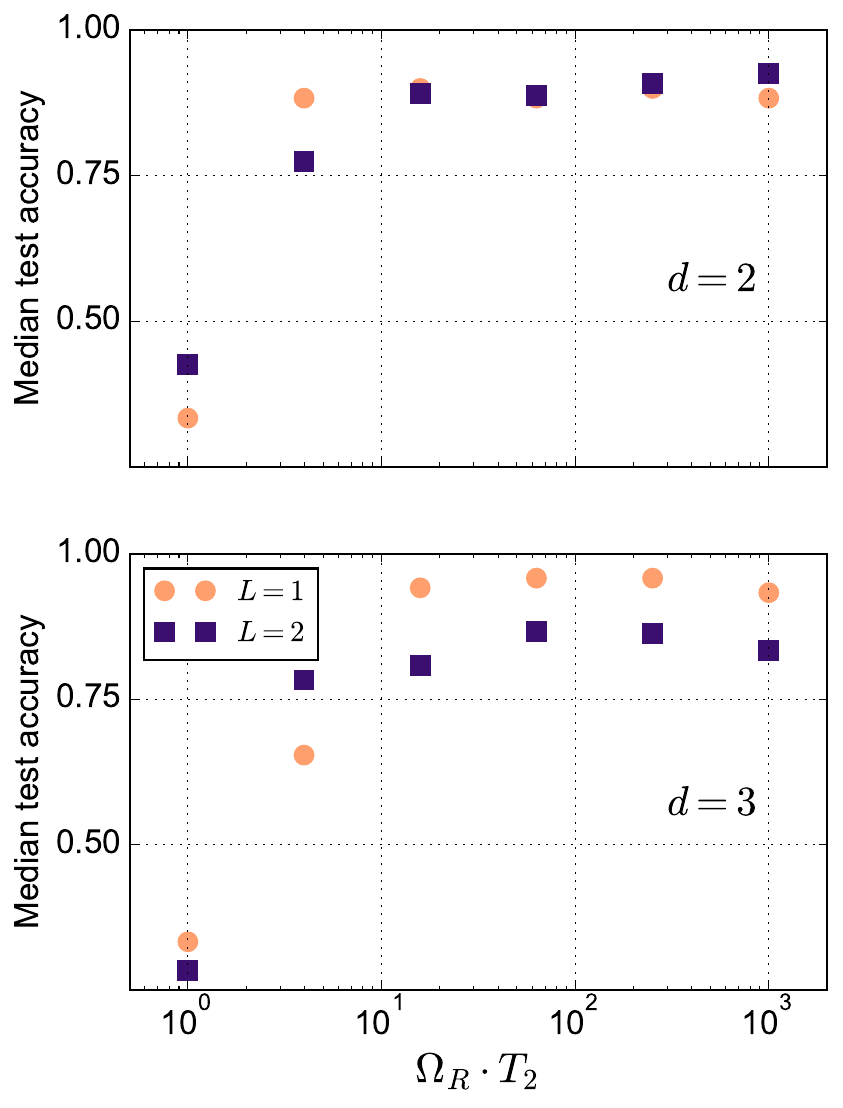}
    \caption{Median of the distribution of test accuracies obtained with the noisy model for different qudit dimensions and number of layers. Top: Qubit results. Adding more layers helps to increase slightly the median value except for some points. Bottom: Qutrit results. Here adding more layers does not provide any benefit in general (obstructs optimizer convergence) and we can see that the median values are higher than in the qubit case.}
    \label{fig:noisesims}
\end{figure}

\section{Conclusions}\label{sec:conclusions}

Challenging the performance of quantum processors in machine learning is an intriguing and timely pursuit. Simple systems not only lend themselves to easier experimental implementation but also facilitate a more comprehensive theoretical understanding of their performance. It is with this motivation that we embarked on our investigation in this paper.
In particular, we focus on a quantum system with $d$ levels, in which transitions between these levels can be triggered through external fields. By optimizing these transitions, we demonstrate that qudits can effectively learn to classify realistic datasets.

We explore both implicit and explicit metric learning paradigms, along with various encoding strategies. Our comprehensive study leads us to the following conclusions: 

\begin{itemize}
  \item Various problems can derive advantages from distinct encoding strategies. This observation aligns with our discussion in section \ref{subsec:EncComp}. In Table \ref{tab:EncodingsG}, we presented a range of functions, and it became evident that the explicit method benefits from the $g_2$ encoding, whereas the implicit method finds better compatibility with the $g_1$ encoding. Moreover, when considering the hybrid model introduced in section \ref{subsec:HDdata}, it becomes apparent that the $g_3$ encoding aligns more effectively with the overall structure of the model.

  \item  Within the metric learning framework, we have explored both explicit and implicit methods in depth, substantiating our exploration with geometric rationale that interlinks the two approaches. Notably, both methodologies exhibit remarkable performance, even rivaling some of the most adept classical algorithms. It is worth noting, however, that of the two, the explicit method emerges as the one most amenable to immediate experimental implementations.

  \item These statements find validation through simulations that incorporate prevalent noise sources found in today's physical devices. Being able to generate any set of MOS thanks to a genetic algorithm developed for this purpose \cite{GApreprint}, we  outline in section \ref{sec:Dissipation} a step-by-step algorithm that has demonstrated robust adaptability, and we are confident in its suitability for diverse experimental implementations that align with the discussed requisites.

  \item Concerning the efficacy of introducing additional levels to the physical system, we observe that while augmenting dimensionality notably enhances performance in certain instances, as illustrated in Figure \ref{fig:HCNNresults}, there are situations, depicted in Figures \ref{fig:ansatzExplicitComparison}, \ref{fig:MethodResults} and \ref{fig:noisesims}, where the impact is less evident. In the latter (Fig. \ref{fig:noisesims} and Sec. \ref{sec:Dissipation}), although the qutrit appears to outperform the qubit when $L = 1$, there appears to be a contradiction in Table \ref{tab:noisyresults}, where the same or better maxima test accuracies are obtained with the qubit for smaller values of $T_2$. Consequently, determining the optimal dimensionality for each specific problem remains an open question, contingent upon the characteristics of the physical system (decoherence rates, speed of operations, etc.) and the nature of the problem itself. Thus, the dimensionality of the quantum system can be viewed as a hyperparameter that is worth tuning and investigating to optimize the learning process.

  \item To conclude, it is worth highlighting that the studies conducted encompass real-world datasets, underscoring the adaptability of the tools discussed here. When confronted with datasets of substantial size like the MNIST, surpassing the capacity of the physical unit, we devised a hybrid algorithm that yields satisfactory performance. Nevertheless, it remains evident that physical units with limited levels such as the ones that we have access to simulate serve as bottlenecks, unable to surpass the efficacy of standard classical methods.
  
\end{itemize}

\section*{Acknowledgements}

The authors thank Adri\'an P\'erez-Salinas for his helpful comments and insights during the preparation of this manuscript. The authors acknowledge  funding from the EU (QUANTERA SUMO and FET-OPEN Grant 862893 FATMOLS), the Spanish Government Grants PID2020-115221GB-C41/AEI/10.13039/501100011033 and TED2021-131447B-C21 funded by MCIN/AEI/10.13039/501100011033 and the EU ``NextGenerationEU''/PRTR, the Gobierno de Arag\'on (Grant E09-17R Q-MAD) and the CSIC Quantum Technologies Platform PTI-001.
This work has been financially supported by the Ministry of Economic Affairs and Digital Transformation of the Spanish Government through the QUANTUM ENIA project call - Quantum Spain project, and by the European Union through the Recovery, Transformation and Resilience Plan - NextGenerationEU within the framework of the ``Digital Spain 2026 Agenda''. J. R-R. acknowledges support from the Ministry of Universities of the Spanish Government through the grant FPU2020-07231. S. R-J.  acknowledges financial support
from Gobierno de Arag\'on through a doctoral fellowship.

\bibliographystyle{unsrtnat}
\bibliography{refs}

\begin{thebibliography}{64}
\bibitem[imp()]{implicitMOS}
These centres are found variationally by minimising the distance with their
  corresponding ensemble of points.
\end{thebibliography}
\appendix

\section{Genetic algorithms for Maximally Orthogonal States}\label{app:MOS}

\begin{figure*}[t!]
    \centering
    \includegraphics[width = \textwidth]{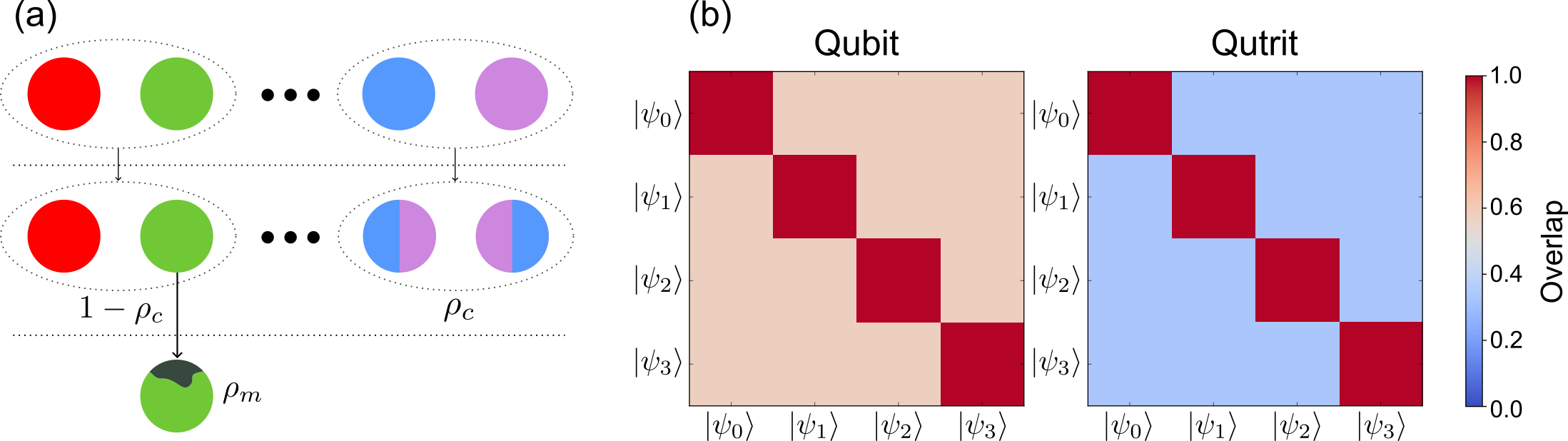}
    \caption{Genetic algorithm to determine maximally orthogonal states whatever the dimension of the qudit might be. (a) Schematic of a generation of the algorithm. In this representation we show the probability of a crossover occurring as $\rho_c$ and the probability of a mutation occurring as $\rho_m$. (b) Example of the determination of 4 maximally orthogonal states in two systems. The colour indicates the intensity of the overlap between each pair of states. As expected, in the qutrit ($d = 3$), having a higher dimension than the qubit, it is easier to accommodate the same number of states while maintaining a smaller overlap.}
    \label{fig:gasketch}
\end{figure*}

As introduced in section \ref{subsec:GeometricInterpretation}, we define the set of maximally orthogonal states (MOS) as the one that optimises a certain function [Cf. Definition \ref{def:mos} in the main text]. The numerical process used is based on bio-inspired algorithms. Specifically, genetic algorithms. The technical details of defining and finding maximally orthogonal states and the algorithm developed for this purpose will be discussed in another publication \cite{GApreprint}. Here we will give a more general overview of the procedure.

A genetic algorithm embodies a versatile optimization approach inspired by the principles of Darwinian natural selection. It involves an evolving population of potential solutions to an optimization problem over successive generations. During each generation, the most adept individuals within the population are chosen. Through processes of mating, mutation, and survival, they give rise to the subsequent generation. This cycle, sketched in Fig. \ref{fig:gasketch} (a), is repeated until a point of convergence is achieved.
The mating process aims to probabilistically amalgamate favorable traits from one individual with those of another, resulting in an overall improved individual. Mutation, on the other hand, introduces randomness, permitting an exploratory journey through the solution space. Survival introduces a deterministic aspect to the algorithm. By allowing the fittest individuals to persist, a consistent progression towards enhanced solutions is ensured, without being hindered by random setbacks.
The challenge of generating maximally orthogonal states constitutes a suitable application of a genetic algorithm. This is due to the straightforward encoding of sets of states as individuals, after which the mating process can be implemented as a recombination of states from the two parent sets.

In our particular setting, for a given qudit dimension $d$, a population, $P_{d, K} = \{\Phi_{d, K}^p\}_{p=1}^N$, will be a collection of sets of $K$ potentially maximally orthogonal states, $\{|\psi_k\rangle \}_{k=1}^K$, the individuals, with $N$ the size of the population.
We seek to maximise the fitness, which is defined in our case as the negative of the energy in Definition \ref{def:mos}. Mating is implemented as a random recombination of the states of two individuals in such a way that the geometric structure is preserved in the process.
In addition to the usual mechanisms of a genetic algorithm, we also apply local optimization to the parents of each generation.

Throughout the main text this algorithm has been used to define the MOS in the case of the MNIST digits dataset, for example. Specifically, in Fig. \ref{fig:gasketch} (b) we show the difference in the results produced by the algorithm between 4 MOS in a qubit and a qutrit. These, in the case of the qubit, are further represented on the Bloch sphere in Fig. \ref{fig:digitsembedding} of the main text.

\section{PCA performance for image classification}\label{app:PCA}

Here we deal with the hand-written digits dataset from the Python package \emph{sklearn}. In Fig. \ref{fig:fiddigits} we compare the performance of the classifier in both training and testing. We also compare the performance when we only try to classify $d$ digits, being $d$ the dimension of the qudit, and when we try to do so for the full dataset. In the first scenario, the best performance is achieved for the qutrit, as it has to deal with a smaller parameter-space dimension. However, for the $d=5$ qudit we see that the fidelities are higher than in the 10 digits case (panel (b)). In the latter, we can see that there is no big difference between the two qudit sizes. The compression of the original data with the PCA technique was different for each size. The final dimension of the data worked with in the algorithm is $2(d-1)$, i.e. 4 in the qutrit and 8 in the $d = 5$ qudit. These low dimensions compared to the original may be the cause of the poor performance of our classifier. That's why in Table \ref{tab:pcadim} we study the dependence of such performance with the PCA dimension. Our test subject is a qutrit trying to classify 5 digits with just one layer.

\begin{figure*}[b!]
    \centering
    \includegraphics[width = 0.9\textwidth]{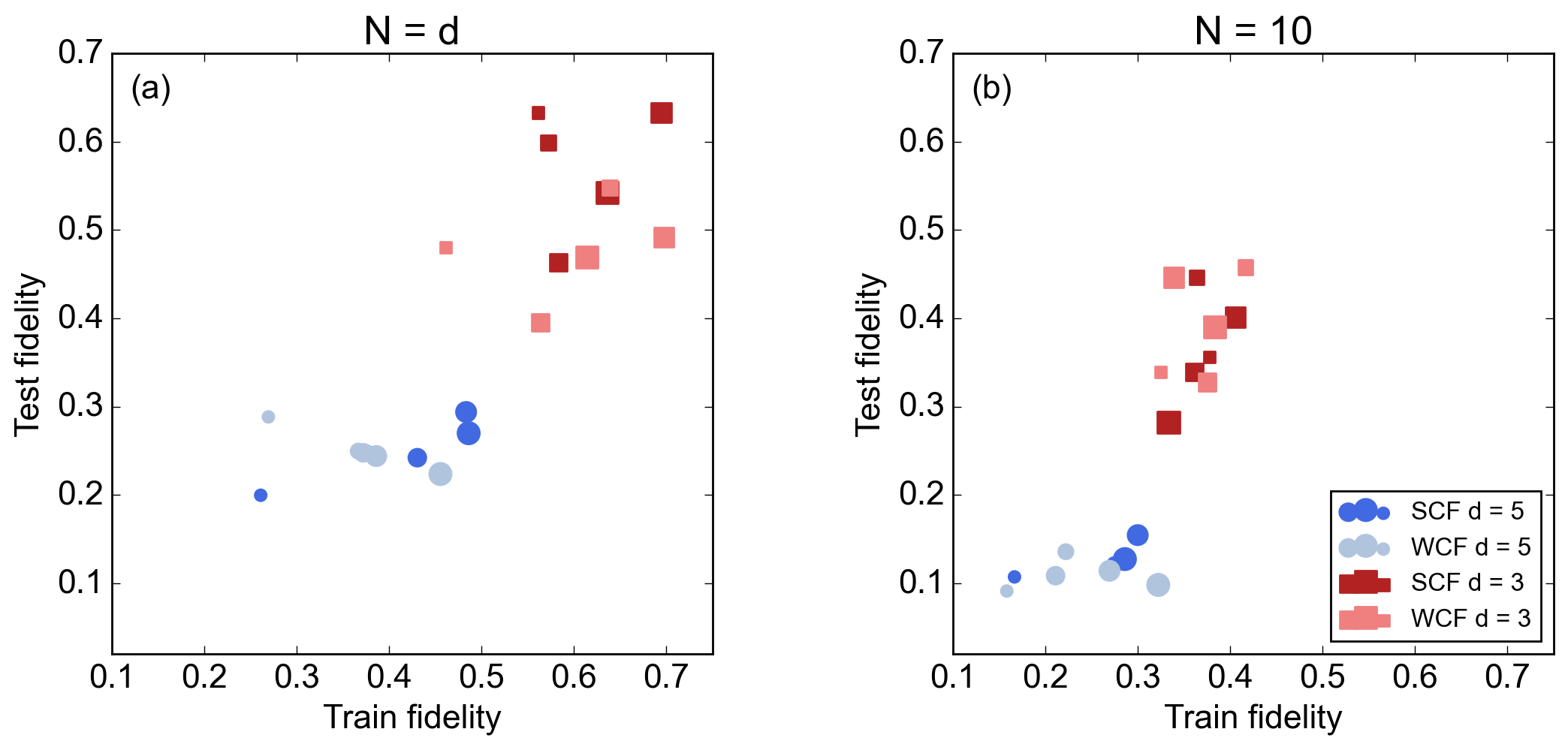}
    \caption{Fidelity obtained for both testing and training with a qutrit and with a d = 5 qudit for different number of layers. In (a) we check the performance for the same number of digits as levels in the qudit and in (b) for the full 10 digits dataset. SCF and WCF stand for Simple Cost Function and Weighted Cost Function, respectively. They refer to the costs functions used in Ref. \cite{perez2020data}. The Simple one is analog to Eq. \eqref{eq:loss_explicit} in the main text. The Weighted one was considered for the sake of completeness. The size of the marker indicates the number of layers employed, from 1 (smallest) to 5 (biggest).}
    \label{fig:fiddigits}
\end{figure*}

\begin{table}[h!]
\centering
\resizebox{0.75\columnwidth}{!}{%
\begin{tabular}{|c|c|c|c|c|c|}
\hline
dim(PCA) & 4    & 6    & 8    & 12   & 18 \\ \hline
Training & 0.30 & 0.33 & 0.39 & 0.55 & 0.53 \\ \hline
Test     & 0.20 & 0.20 & 0.23 & 0.24 & 0.21 \\ \hline
\end{tabular}
}
\caption{Fidelities obtained with a one layer qutrit ansatz classifying 5 digits as a function of the PCA's dimension.}
\label{tab:pcadim}
\end{table}

\end{document}